\documentclass[11pt]{article}
\usepackage{graphicx} 
\usepackage[utf8]{inputenc}
\usepackage{eurosym,geometry,graphicx,color,setspace,sectsty,comment,footmisc,natbib,pdflscape,array,hyperref}
\usepackage[normalem]{ulem}
\usepackage{caption}
\usepackage{subcaption}
\usepackage{amssymb,amsfonts,amsmath}
\usepackage{amsthm}
\usepackage{natbib}
\usepackage{booktabs}
\usepackage{threeparttable}
\usepackage{tikzsymbols} 
\usepackage{tikz}
\usetikzlibrary{patterns,arrows}
\usepackage{apxproof}
\usepackage{siunitx}
\usepackage[capposition=top]{floatrow} 

\usepackage{xcolor} 
\usepackage{todonotes}
\setuptodonotes{inline}
\usepackage[]{mdframed}

\newtheorem{lemma}{Lemma}

\newtheorem{proposition}{Proposition}

\theoremstyle{definition} 

\newtheorem{remark}{Remark}

\newtheorem{example}{Example}
\newtheorem*{example*}{Example}

\newtheorem*{fact*}{Fact}

\AtBeginEnvironment{example}{%
  \pushQED{\qed}%
}
\AtEndEnvironment{example}{\popQED\endexample}

\hypersetup{colorlinks,citecolor=blue} 
\title{Concentration-Based Inference for Evaluating Horizontal Mergers\footnote{I thank Dan Greenfield, Devesh Raval, Ted Rosenbaum, David Schmidt, Charles Taragin, Shawn Ulrick, and Brett Wendling for their helpful comments. Much of this article was written during the author’s tenure at the U.S. Federal Trade Commission; the views expressed in this paper are those of the author and do not necessarily represent those of the Federal Trade Commission or any of its Commissioners. All errors are mine.}}
\author{Paul S. Koh\footnote{School of Economics, Yonsei University. Seoul, Republic of Korea. Email: \texttt{paulkoh9@gmail.com}.}}
\date{September 20, 2025}

\begin{document}

\maketitle
\begin{abstract}
Antitrust authorities routinely rely on market concentration measures to assess the potential adverse effects of mergers on consumer welfare. Using a first-order approximation argument with logit and CES demand, I derive the relationship between the welfare effect of a merger on consumer surplus and the change in the Herfindahl-Hirschman Index (HHI). My results suggest that merger harm is correlated with the merger-induced change in HHI, and the proportionality coefficient depends on the price responsiveness parameter, market size, and the distribution of market shares within and across the merging firms. I present numerical validation of my formula along with an empirical illustration.

    \vspace{1em}
\end{abstract}

\clearpage
\section{Introduction} 
Concentration measures have played a critical role during merger enforcement at the screening stage and in court proceedings. In the United States, \emph{Merger Guidelines}, jointly drafted by the Department of Justice and the Federal Trade Commission, state how the agency uses the pre-merger level of Herfindahl-Hirschman index (HHI) and merger-induced changes in HHI (often denoted $\Delta$HHI) to presume anticompetitive mergers. 

Previous works in the merger analysis literature have recognized the positive correlation between the changes in HHI and unilateral effects in both homogeneous and differentiated products settings \citep*{kim1993mergers, froeb1998robust, prager1998substantial, shapiro20102010, ashenfelter2015efficiencies, miller2017upward, miller2021quantitative, spiegel2021herfindahl, bhattacharya2023merger}.\footnote{In particular, \citet{shapiro20102010} shows that, with logit demand, the sum of the approximate diversion ratios is $s_j  +s_k + 2s_js_k$ or $s_j + s_k + \Delta HHI$. My derivation is different from his. I also focus on the welfare effects of mergers.} More recently, \citet{nocke2022concentration} has shown that compensating synergies that offset market power effects of mergers can be expressed in terms of $\Delta$HHI in both the Cournot and Bertrand cases. \citet{nocke2023aggregative} shows how $\Delta$HHI relates to merger welfare effects in differentiated products settings. 

This paper contributes to the broad literature on concentration measures and market power by providing a novel characterization of how HHI measures relate to merger welfare effects in multiproduct oligopoly settings with logit and CES demand. Using the first-order approach arguments \`{a} la \citet{jaffe2013first}, I show 
\begin{equation}\label{equation:what.I.show}
    \Delta \mathit{CS} \approx - \rho \Delta \mathit{HHI},
\end{equation}
where $\Delta \mathit{CS}$ is the change in consumer surplus, $\Delta \mathit{HHI}$ is the change in HHI computed with pre-merger market shares (in quantity if the demand is logit and in revenue if CES), and $\rho$ is the proportionality coefficient whose value depends on the market size, price responsiveness parameter, and the distribution of market shares across and within the merging firms.\footnote{The first-order approach to merger analysis has gained significant popularity due to its tractability; see, e.g., \citet{farrell2010antitrust}, \citet{valletti2021mergers}, and \citet{koh2024merger}.} Equation \eqref{equation:what.I.show} provides a simple but novel angle for understanding how mergers with the same value of $\Delta$HHI may differ in their impact on consumer welfare based on prior knowledge of industry statistics (i.e., market size, price elasticity, and the distribution of firm sizes). 

Formula \eqref{equation:what.I.show} provides a simple tool for connecting $\Delta$HHI to predicted consumer harm at the merger screening stage. I show how the proportionality coefficient $\rho$ can be derived directly from the merging firms' product-level shares and the price responsiveness parameter, including how to compute the merger pass-through matrix under both logit and CES demand systems. I build on \citet{miller2017upward}'s insight that when estimating merger pass-through rates is difficult, it is reasonable to approximate the merger pass-through matrix $M$ with an identity matrix. Unlike \citet{miller2017upward}, however, I show that the approximation is theoretically valid and can be refined when the demand is CES. Specifically, I prove that $M \to I$ under logit demand and $M \to \frac{\sigma}{\sigma - 1}I$ under CES demand as the merging firms' market shares approach zero, where $\sigma$ is the elasticity of substitution. These results provide new theoretical support for the approximation by \citet{miller2017upward}, and I present numerical evidence showing its accuracy and practical usefulness.

The closest work to this paper is \citet{nocke2023aggregative}, which also shows \eqref{equation:what.I.show} holds using second-order Taylor expansions around small shares and monopolistic competition conduct. However, their version of $\rho$ depends only on the market size and price responsiveness parameter. My first-order approach-based derivation, which relies on a different Taylor approximation employed by the first-order approach literature, shows how $\rho$ also depends on the \emph{distribution} of merging firms' shares.\footnote{To be clear, \citet{jaffe2013first}'s first-order Taylor approximation argument relies on a small upward pricing pressure assumption but not a small market share assumption.} Monte Carlo simulation shows that my approach can improve the prediction power relative to \citet{nocke2023aggregative}'s approach. My results corroborate
the findings of \citet{nocke2022concentration} and \citet{nocke2023aggregative} that the unilateral effects of horizontal mergers are directly linked to the change in HHI but also demonstrate that pre-merger market shares may still retain relevance even after conditioning on the change in HHI.

The rest of the paper is organized as follows. Section \ref{section:model} describes the model and reviews the first-order approach to merger analysis. Section \ref{section:consumer.surplus.and.hhi} presents the main proposition that characterizes $\rho$ in \eqref{equation:what.I.show} as a simple function of market size, price responsiveness parameter, and the distribution of shares within and across firms. Section \ref{section:analysis.of.the.proportionality.coefficient} analyzes the behavior of the proportionality coefficient $\rho$ with respect to model primitives and merging firms' market shares. Section \ref{section:comparison.to.Nocke.Shutz} compares my results to those of \citet{nocke2023aggregative}. Section \ref{section:empirical.example} considers an empirical application to the Heinz/Beech-Nut merger to demonstrate the usefulness of my framework. Section \ref{section:conclusion} concludes. All proofs are in Appendix \ref{section:proofs}.
\section{Model\label{section:model}}

\subsection{Setup}
I consider a Bertrand-Nash pricing game of multiproduct firms with differentiated products.\footnote{Here, ``multiproduct'' means that a firm may offer multiple product varieties rather than conglomerate-type products and sales.} Each firm $f$'s profit function is given by
\begin{equation}
    \Pi_f(p) = \sum_{l \in \mathcal{J}_f} (p_l - c_l) q_l(p),
\end{equation}
where $p_l$, $c_l$, and $q_l(p)$ represent product $l$'s price, (constant) marginal cost, and quantity demanded given price vector $p$, respectively; $\mathcal{J}_f$ is the set of products produced by firm $f$.

Following \citet{nocke2018multiproduct}, I consider multinomial logit (MNL) and constant elasticity of substitution (CES) demand functions, under which the demand for product $j$ is given by 
\begin{equation}\label{equation:demand.functions}
q_j(p) = 
\begin{cases}
    \frac{\exp(v_j - \alpha p_j)}{1 + \sum_{l \in \mathcal{J}} \exp(v_l - \alpha p_l)} N & \text{(MNL)} \\
    \frac{v_j p_j^{-\sigma}}{1 + \sum_{l \in \mathcal{J}} v_l p_l^{1-\sigma}} Y & \text{(CES)}
\end{cases} 
,
\end{equation}
where $N$ and $\alpha>0$ (resp. $Y$ and $\sigma>1$) are the market size and price responsiveness parameters in the MNL (resp. CES) model, respectively. In typical applications, $N$ and $Y$ represent the total number of purchase opportunities (e.g., annual units sold) and the total category expenditure (e.g., annual sales) in the relevant product/geographic market.\footnote{I treat market size as exogenous. While the hypothetical monopolist test (HMT) may play a role in delineating market boundaries and size, my approach can be more practically applied to candidate market definitions at the screening stage. See, for example, \citet{koh2024market} for a sensitivity analysis framework regarding market definition. } In each model, $v_j$ is the term that captures product $j$'s quality. Note that the denominators in \eqref{equation:demand.functions} also reflect the standard normalization of the value of the outside option.\footnote{I review \citet{nocke2018multiproduct}'s aggregative games framework in Online Appendix \ref{section:review.of.aggregative.games.framework} for readers less familiar with the framework. I also refer readers to \citet{caradonna2024mergers} Online Appendix A for an excellent review of the aggregative games framework.}

\subsection{First-Order Approximation of Merger Welfare Effects}
I approximate the impact of mergers on consumer surplus using the first-order approach of \citet{jaffe2013first}. This method utilizes pre-merger data to calculate margins and diversion ratios, translates them into predicted price changes using a merger pass-through matrix, and then infers the resulting effects on consumer welfare. Formally, the process unfolds as follows.

Suppose firms $A$ and $B$ merge. The upward pricing pressure (UPP) of product $j \in \mathcal{J}_A$ is defined as
\begin{equation}\label{equation:upward.pricing.pressure}
    \mathit{UPP}_j \equiv \sum_{l \in \mathcal{J}_B} (p_l - c_l) D_{j \to l},
\end{equation}
where $D_{j \to l} \equiv - (\partial q_l / \partial p_j) / (\partial q_j / \partial p_j)$ is the quantity diversion ratio from product $j$ to $l$ \citep{farrell2010antitrust}. The definition of upward pricing pressure is symmetric for products $j \in \mathcal{J}_B$. 

\citet{jaffe2013first} show that predicting merger-induced price changes from upward pricing pressures (net of efficiencies) requires a \emph{merger pass-through matrix}, which captures how these pressures spread across prices. They derive the first-order approximation $\Delta p \approx M \cdot \mathit{UPP}$, where $M$ is the merger pass-through matrix and $\mathit{UPP}$ is the vector of upward pricing pressures for the merging firms' products.\footnote{Formally, the merger pass-through matrix is defined as $M = - (\partial h(p^\text{pre}) / \partial p)^{-1}$, where $h(p^\text{post}) = 0$ characterizes the post-merger profit maximization conditions, with each $h_j$ normalized to be quasilinear in marginal cost of product $j$. As observed by \citet{koh2024merger}, in the case of CES demand, it is more convenient to work with the merger pass-through matrix defined as $M = - (\partial h(\tilde{p}^\text{pre}) /\partial \tilde{p})^{-1} $, where $\tilde{p}_j \equiv \log p_j$. This in turn yields $\%\Delta p \approx M \cdot \mathit{GUPPI}$, where $(\% \Delta p)_j = \Delta p_j/p_j$ is the percentage change in price of product $j$ and $\mathit{GUPPI}_j \equiv \mathit{UPP}_j / p_j$. }

Once the analyst obtains these UPP-based predicted price changes, the total impact on consumer surplus can be approximated by summing across products:
\[
\Delta \mathit{CS} \approx \sum_{j \in \mathcal{J}_A \cup \mathcal{J}_B} \Delta \mathit{CS}_j,
\]
where $\Delta \mathit{CS}_j \approx - \Delta p_j \times q_j$ for all $j \in \mathcal{J}_A \cup \mathcal{J}_B$. Thus, this framework connects post-merger price effects to the underlying diversion ratios and margins, which in turn reflect the merging firms’ pre-merger market shares.

\subsection{Small-Share Approximation of Merger Pass-Through Matrix}

The accuracy of the first-order approach depends on the availability of the merger pass-through matrix $M$, which is often difficult to estimate due to its reliance on demand curvature. While $M$ can be expressed using market shares and the price responsiveness parameter under logit or CES demand, the required matrix inversion makes closed-form solutions unwieldy.\footnote{See Online Appendix \ref{section:merger.pass.through.matrix.under.logit.and.ces} for a full characterization of the merger pass-through matrix under logit and CES demand.} The next lemma offers a key simplification, allowing the first-order approach to be applied without explicitly calculating $M$.

\begin{lemma} \label{lemma:small.share.approximation.of.merger.pass.through.matrix}
If the merging firms' market shares approach zero, the merger pass-through matrix converges to $M \to \varphi I$, where $\varphi = 1$ under MNL demand and $\varphi = \frac{\sigma}{\sigma - 1}$ under CES demand.
\end{lemma}
\begin{proof}
    See Appendix \ref{section:proof.of.lemma.1}.
\end{proof}

Lemma \ref{lemma:small.share.approximation.of.merger.pass.through.matrix} provides theoretical support for the approximation $M \approx I$ proposed by \citet{miller2017upward}, who use simulation evidence to show that the diagonal elements of the pass-through matrix generally dominate the off-diagonal terms.\footnote{\citet{miller2017upward}'s approximation $\Delta p_j \approx \mathit{UPP}_j$ has been proven helpful for practitioners because pass-through rates are usually difficult to estimate \citep{farrell2010antitrust, jaffe2013first, weyl2013pass, miller2016pass, miller2017upward, miller2017pass, koh2024merger}.} My results refine their approach by showing that under CES demand, this approximation should be scaled by a factor of $\frac{\sigma}{\sigma - 1} > 1$. This adjustment reflects that CES demand typically implies higher pass-through rates than logit demand, allowing the Miller et al. approximation to be tailored more accurately to the CES setting.

Lemma \ref{lemma:small.share.approximation.of.merger.pass.through.matrix} suggests that it is reasonable to approximate each merger price effect as $\Delta p_j \approx \varphi \mathit{UPP}_j$. To assess the performance of this approach, I repeat the Monte Carlo simulation exercise from \citet{miller2017upward} and evaluate its accuracy.\footnote{I explain the simulation details in Appendix \ref{section:simulation}. Note that \citet{miller2017upward} does not study the case of CES demand. The result for logit is consistent with \citet{miller2017pass}; see their Figure 2.} Figure \ref{figure:accuracy.of.upp} shows that $\varphi \mathit{UPP}_j$ closely predicts the actual merger price effects $\Delta p_j$ when the UPPs are small, which typically occurs when the merging firms have small market shares. Under logit demand, the approximation remains accurate across a wide range of UPP values. In contrast, under CES demand, it tends to understate the actual price effects as UPP increases, consistent with the findings of \citet{miller2017upward}.\footnote{In the case of CES demand, the predicted merger price effects may be conservative, likely because focusing only on first-order effects overlooks the strong iterative effects that arise under convex demand.} Overall, Figure \ref{figure:accuracy.of.upp} indicates that using $M \approx \varphi I$ offers a straightforward and practical way to apply the first-order approach without directly estimating the merger pass-through matrix.

\begin{figure}[htbp!]
\centering
\begin{subfigure}{.5\textwidth}
  \centering
  \includegraphics[width=.6\linewidth]{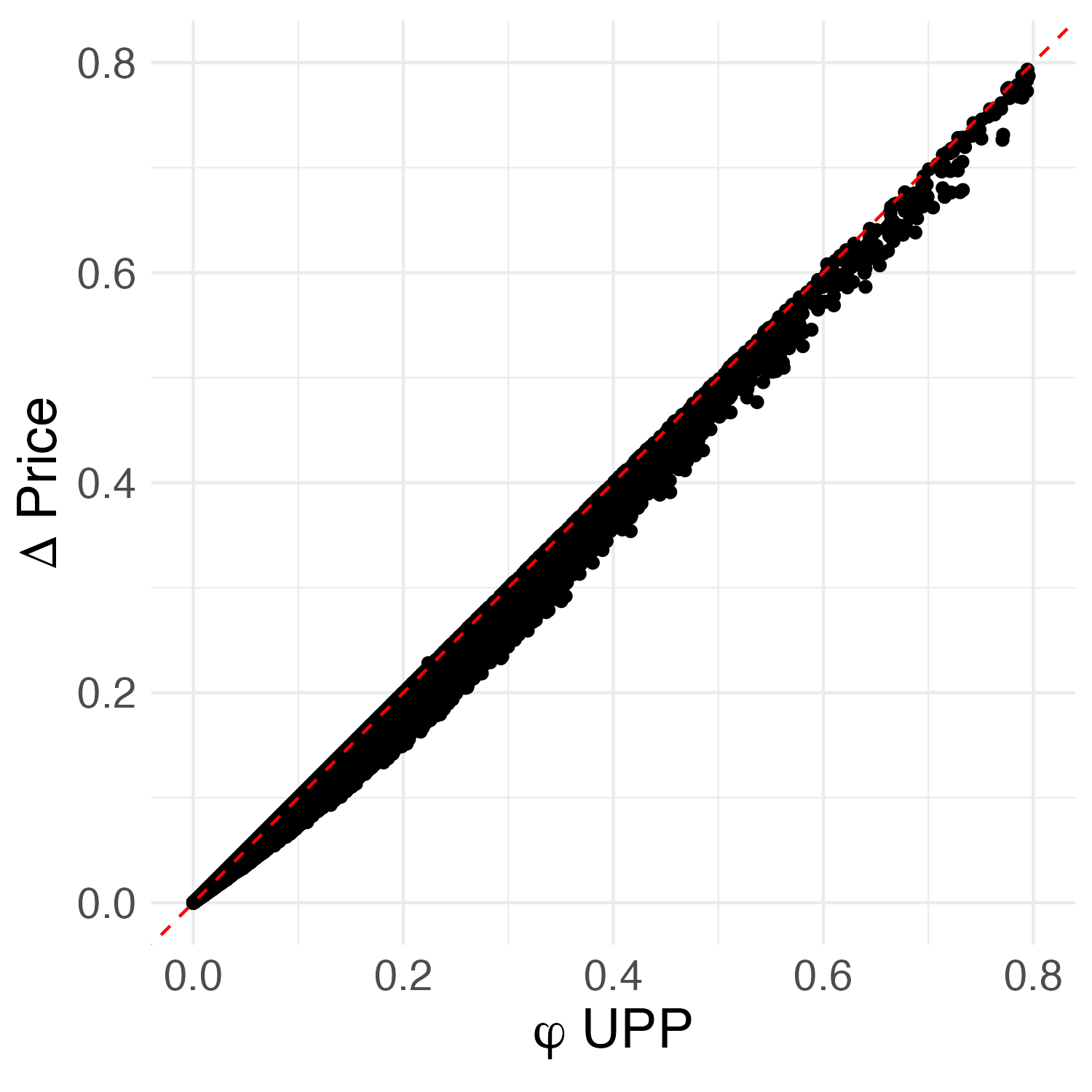}
  \caption{Logit ($\varphi = 1$)}
  \label{fig:sub1}
\end{subfigure}%
\begin{subfigure}{.5\textwidth}
  \centering
  \includegraphics[width=.6\linewidth]{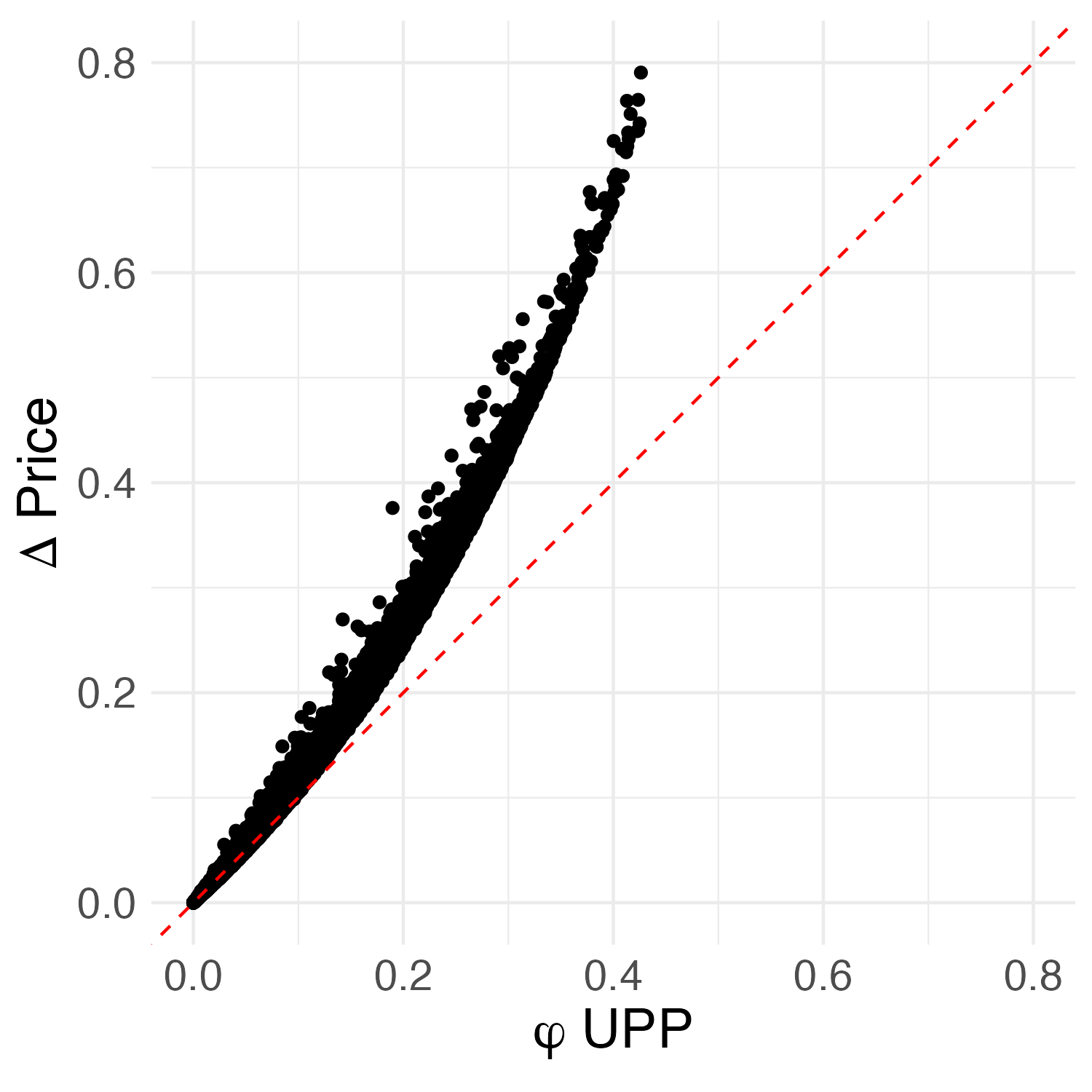}
  \caption{CES ($\varphi = \frac{\sigma}{\sigma - 1}$)}
  \label{fig:sub2}
\end{subfigure}
\caption{Accuracy of the $\Delta p_j \approx \varphi \mathit{UPP}_j$ Approximation}
\label{figure:accuracy.of.upp}
\floatfoot{\emph{Notes}: The figures compare the accuracy of the $\varphi \mathit{UPP}_j$ against the true price change $\Delta p_j$ simulated from mergers following the Monte Carlo simulation design of \citet{miller2017upward}.}
\end{figure}

\subsection{Limitations}

The first-order approach to merger analysis has important limitations. First, by construction, it only captures the direct response of the merging firms to the merger, without accounting for equilibrium feedback effects between them or allowing for strategic reactions by non-merging rivals. Second, it also focuses exclusively on unilateral effects, abstracting from the possibility of coordinated conduct. Finally, estimating merger pass-through rates may be challenging in practice; although Lemma \ref{lemma:small.share.approximation.of.merger.pass.through.matrix} provides a useful alternative, it can generate bias when the merging firms' shares are large.\footnote{\citet{miller2017upward} also shows that the first-order approach may underpredict merger price effects when the demand function is convex, as in the case of CES demand. The first-order approach may overpredict merger price effects when prices are strategic substitutes, which can arise in random coefficient logit models, but such examples are not well-documented in the literature. }

Its main advantage, however, lies in its simplicity: it provides a sufficient-statistics framework that links pre-merger information—such as market concentration measures in this paper—to predictions about post-merger price effects and consumer harm. In the following sections, I show that this framework offers a clear characterization of how changes in the HHI translate into losses in consumer surplus. 
\section{Consumer Surplus and HHI \label{section:consumer.surplus.and.hhi}}

\subsection{Review of HHI}
The Herfindahl-Hirschman Index (HHI) is defined as
\[
\mathit{HHI} = \sum_{f \in \mathcal{F}} s_f^2.
\]
Here, $s_j$ denotes the share of product $j$, which may be based on either quantity or revenue. The total share of firm $f$ is given by $s_f = \sum_{j \in \mathcal{J}_f} s_j$, where $J_f$ is the set of products sold by firm $f$. I use $s_j^Q$ and $s_j^R$ to distinguish between quantity and revenue shares, respectively. Under the logit framework, $s_j^Q = q_j / N$. Under the CES framework, $s_j^R = p_jq_j / Y$.\footnote{The shares are calculated with the outside option. Shares excluding the outside option can be computed as $\tilde{s}_f = s_f/(1-s_0)$. }  Note that I keep the HHI scaled between 0 and 1, rather than multiplying by 10,000 as is commonly done to improve readability.

If firms $A$ and $B$ merge, the resulting change in the HHI is given by
\[
\Delta \mathit{HHI}_{AB} = 2s_A s_B.
\]
This calculation is based on pre-merger market shares and does not account for any shifts in market shares that may occur in the post-merger equilibrium.\footnote{Specifically, $\mathit{HHI}^\text{post} - \mathit{HHI}^\text{pre} = ((s_A+s_B)^2 + \sum_{f \neq A,B} s_f^2) - (s_A^2 + s_B^2 + \sum_{f \neq A,B} s_f^2) = 2s_A s_B$.} However, as I show below, using pre-merger shares in this way is logically consistent with the first-order approach, which relies solely on pre-merger information to predict potential merger harm.

\subsection{Relationship between Consumer Surplus and HHI}

The following proposition, which is the main result of this paper, shows that $\Delta \mathit{HHI}_{AB}$ is closely related to the impact of the merger on consumer surplus predicted by the first-order approach. 

\begin{proposition}\label{proposition:consumer.surplus.first.order.approach} 
Suppose firms $A$ and $B$ merge, and there is no merger-specific synergy. The first-order approximation to the merger-induced change in consumer surplus is 
\begin{equation}
\Delta \mathit{CS} =
\begin{cases}\label{equation:delta.cs.and.delta.hhi}
    -V_0^{\mathit{MNL}} \rho_1^{\mathit{MNL}} \rho_2^{\mathit{MNL}} \Delta \mathit{HHI}_{AB}^Q & (\mathit{MNL}) \\
    -V_0^\mathit{CES} \rho_1^\mathit{CES} \rho_2^\mathit{CES} \Delta \mathit{HHI}_{AB}^R & (\mathit{CES})
\end{cases}
,    
\end{equation}
where
\begin{align}
    V_0 & \equiv 
    \begin{cases}
        \frac{N}{\alpha} & \text{(MNL)} \\
        \frac{Y}{\sigma - 1} & \text{(CES)}
    \end{cases} \label{equation:monetary.scaling.factor}, \\
    \rho_1 & \equiv 
    \begin{cases}
        \frac{1}{\left(1-s_A\right) \left(1-s_B \right)} & \text{(MNL)} \\
        \frac{\frac{\sigma}{\sigma - 1}}{\left( \frac{\sigma}{\sigma - 1} -s_A^R \right) \left( \frac{\sigma}{\sigma - 1} -s_B^R \right)} & \text{(CES)} \\
    \end{cases}
    \label{equation:cross.firm.scaling.factor}
    ,
    \\
    \rho_2 & \equiv 
    \begin{cases}
         \sum_{j,l \in \mathcal{J_A} \cup \mathcal{J}_B} M_{jl} \frac{s_j}{s_l} \left(\frac{1}{2} \mathbb{I}_{l \in \mathcal{J}_A} \frac{s_l/(1-s_l)}{s_A/(1-s_A)} + \frac{1}{2} \mathbb{I}_{l \in \mathcal{J}_B} \frac{s_l/(1-s_l)}{s_B/(1-s_B)} \right) & \text{(MNL)} \\
         \sum_{j,l \in \mathcal{J_A} \cup \mathcal{J}_B} \frac{\sigma - 1}{\sigma} M_{jl} \frac{s_j^R}{s_l^R} \left(\frac{1}{2} \mathbb{I}_{l \in \mathcal{J}_A} \frac{s_l^R/(\frac{\sigma}{\sigma-1}-s_l^R)}{s_A^R/(\frac{\sigma}{\sigma-1}-s_A)} + \frac{1}{2} \mathbb{I}_{l \in \mathcal{J}_B} \frac{s_l^R/(\frac{\sigma}{\sigma-1}-s_l^R)}{s_B^R/(\frac{\sigma}{\sigma-1}-s_B^R)} \right) & \text{(CES)}
    \end{cases}
    \label{equation:within.firm.scaling.factor}
    .
\end{align}
\end{proposition}
\begin{proof}
    See Appendix \ref{section:proof.of.proposition.1}.
\end{proof}

Proposition \ref{proposition:consumer.surplus.first.order.approach} shows that $\Delta \mathit{CS}$ is proportional to $\Delta \mathit{HHI}_{AB}$, with the proportionality coefficient determined by the market size, price responsiveness parameter, and the distribution of market shares across and within firms (the latter relevant when merging firms are multiproduct producers). Thus, equation \eqref{equation:delta.cs.and.delta.hhi} provides a means to directly translate $\Delta \mathit{HHI}_{AB}$ into statements about harm to consumers using data from the merging firms. 

First, $V_0$ converts $\Delta \mathit{HHI}_{AB}$ into monetary units, as it is the product of market size and the inverse of the price responsiveness parameter. Second, $\rho_1$ serves as a cross-firm scaling factor, reflecting the interaction of total shares across firms. Finally, $\rho_2$ functions as a within-firm scaling factor, capturing how shares are distributed among products within multiproduct firms. 

I defer further interpretation and comparative statics analysis for the components of the proportionality coefficients to Section \ref{section:analysis.of.the.proportionality.coefficient}. In Section \ref{section:comparison.to.Nocke.Shutz}, I compare equation \eqref{equation:delta.cs.and.delta.hhi} with the formula derived by \citet{nocke2023aggregative}. Finally, I provide an empirical example in Section \ref{section:empirical.example} using the Heinz/Beech-Nut merger in the U.S. market for baby food.

\begin{remark}[Compact Representation]
Recalling $\varphi = 1$ if MNL and $\varphi = \frac{\sigma}{\sigma - 1}$ if CES, the first-order approximation to the merger welfare effect to consumers \eqref{equation:delta.cs.and.delta.hhi} can be compactly written as 
\begin{align*}
\Delta \mathit{CS}  & = -V_0 \rho_1 \rho_2 \Delta \mathit{HHI}, \quad \text{where }  \\
    \rho_1  &\equiv \frac{\varphi}{(\varphi - s_A)(\varphi - s_B)}, \quad \text{and} \\
    \rho_2  & \equiv \sum_{j,l\in \mathcal{J}_A \cup \mathcal{J}_B} \left( \frac{M_{jl}}{\varphi} \right) \left( \frac{s_j}{s_l} \right) \left( \frac{1}{2} \mathbb{I}_{l \in \mathcal{J}_A } \frac{s_l/(\varphi-s_l)}{s_A/(\varphi - s_A)} + \frac{1}{2} \mathbb{I}_{l \in \mathcal{J}_B} \frac{s_l/(\varphi - s_l)}{s_B/(\varphi - s_B)} \right).
\end{align*}    
\end{remark}

\begin{remark}[Single-Product Case] \label{remark:single.product.case}
    If the merging firms are producers of single products $\rho_2$ boils down to $\rho_2 = \frac{1}{2\varphi}(M_{AA} + M_{AB}\frac{s_A}{s_B} + M_{BA} \frac{s_B}{s_A} + M_{BB})$. Further assuming $M \approx \varphi I$ simplifies this expression to $\rho_2 \approx 1$. In this case, expression \eqref{equation:delta.cs.and.delta.hhi} reduces to $\Delta \mathit{CS} \approx - V_0 \left( \frac{\varphi}{(\varphi - s_A)(\varphi - s_B)} \right) \Delta \mathit{HHI}_{AB}$, which provides a simple and direct way to connect $\Delta \mathit{HHI}_{AB}$ to consumer harm based on firm-level shares.
\end{remark}

\begin{remark}[Merger Screening Decision Rule]
In the spirit of the first-order approach, equation \eqref{equation:delta.cs.and.delta.hhi} allows the analyst to estimate total consumer harm using only data from the merging firms. If the agency seeks to evaluate consumer harm net of potential merger-specific cost savings, the formula can be adjusted to incorporate marginal cost reductions.\footnote{The net upward pricing pressure for a product $j \in \mathcal{J}_A \cup \mathcal{J}_B$ is $\mathit{UPP}_j^\text{net} = \mathit{UPP}_j - \Delta c_j$, where $\Delta c_j = c_j^\text{pre} - c_j^\text{post} > 0$ is the savings from marginal cost. Then since $\Delta p^\text{net} \approx M \cdot \mathit{UPP}^\text{net} = M \cdot (\mathit{UPP} - \Delta c)$, we get $\Delta \mathit{CS}^\text{net} \approx - \Delta p^\text{net$\top$} q \approx -\mathit{UPP}^\top M^\top q + \Delta c^\top M^\top q$. Since, $\eqref{equation:delta.cs.and.delta.hhi}=-\mathit{UPP}^\top M^\top q$, accounting for potential cost efficiency amounts to adding back $\Delta c^\top M^\top q >0$ to \eqref{equation:delta.cs.and.delta.hhi}. Using $M \approx \varphi I$, the cost savings term can be further simplified as $\Delta c^\top M^\top q \approx \varphi \sum_{j \in \mathcal{J}_A \cup \mathcal{J}_B} \Delta c_j q_j$. } This provides a clear decision rule: the agency should block the proposed merger if the net impact on consumer surplus is negative. Unlike the traditional screening rule, which relies solely on the level and change in HHI, this approach explicitly incorporates efficiency considerations, offering a more nuanced and economically grounded assessment of merger effects.     
\end{remark}



\section{Analysis of the Proportionality Coefficient \label{section:analysis.of.the.proportionality.coefficient}}
Two mergers with the same value of $\Delta$HHI may have different proportionality coefficients due to differences in the distribution of shares of the merging firms. To sharpen the intuition behind how the proportionality coefficients behave, I analyze how the proportionality coefficient in \eqref{equation:delta.cs.and.delta.hhi} depends on the market size, consumer price responsiveness, and the distribution of shares across and within firms.

\subsection{Monetary Scaling Factor}
The term $V_0$ in \eqref{equation:monetary.scaling.factor} serves as a monetary scaling factor that translates $\Delta \mathit{HHI}_{AB}$ into monetary units. It equals the product of market size and the inverse of the price responsiveness parameter. Under the logit model, the price responsiveness parameter $\alpha$ represents the marginal utility of income, so $1/\alpha$ converts utility units to monetary units (e.g., dollars), and multiplying by $N$ scales the per-capita measure to the total market. In the CES case, the marginal utility of income is $(\sigma - 1)/Y$, so multiplying by $Y/(\sigma - 1)$ converts welfare changes in utils to monetary units.\footnote{The indirect CES utility is $V(p,Y) = Y/P(p)$, where $P(p)$ is the CES price index. Let $u = (\sigma - 1)\log V(p,Y)$ be the log-cardinalized utility scaled by $(\sigma-1)$. The marginal utility of income is then $\partial u/\partial Y = (\sigma - 1)/Y$. } Viewed this way, $-\rho_1 \rho_2 \Delta \mathit{HHI}_{AB}$ can be interpreted as the change in consumer utility caused by the merger.

Comparative statics with respect to $V_0$ are straightforward: consumer harm increases with market size and with less price-sensitive demand. In particular, more inelastic demand means consumers offer weaker resistance to post-merger price increases, shifting a greater share of the burden onto them.

\subsection{Cross-Firm Scaling Factor}

The term $\rho_1$ in \eqref{equation:cross.firm.scaling.factor} serves as a cross-firm scaling factor, reflecting the interaction of total shares across firms. The following proposition studies the behavior of $\rho_1$ with respect to the merging firms' firm-level shares.

\begin{proposition}\label{proposition:cross.firm.scaling.factor.1}
Suppose that $s_f < \varphi$ for $f=A,B$.
\begin{enumerate}
    \item \label{proposition:cross.firm.scaling.factor.1.3} $\rho_1 \to 1 / \varphi$ as the merging firms' shares approach zero.
    \item \label{proposition:cross.firm.scaling.factor.1.1} $\rho_1$ is increasing in the merging firms' shares.
    \item \label{proposition:cross.firm.scaling.factor.1.2} $\rho_1$ is strictly convex in the merging firms' shares.
\end{enumerate}
\end{proposition}
\begin{proof}
    See Appendix \ref{section:proof.of.proposition.2}.
\end{proof}

Proposition \ref{proposition:cross.firm.scaling.factor.1}.\ref{proposition:cross.firm.scaling.factor.1.3} implies that approximation $\rho_1 \approx 1/\varphi$ is valid when the merging firms' shares are small; this result will be of particular interest when comparing my formula to that of \citet{nocke2023aggregative} in Section \ref{section:comparison.to.Nocke.Shutz}. Proposition \ref{proposition:cross.firm.scaling.factor.1}.\ref{proposition:cross.firm.scaling.factor.1.1} provides a direction of how $\rho_1$ behaves as merging firms' shares get larger. 

Proposition \ref{proposition:cross.firm.scaling.factor.1}.\ref{proposition:cross.firm.scaling.factor.1.2} implies that a merger between two equally sized firms is less harmful than one between a dominant firm and a much smaller rival \emph{if the two mergers produce the same level of $\Delta \mathit{HHI}_{AB}$}. In an asymmetric merger, the larger firm raises its price on a large base of sales, while the smaller firm raises its price sharply on a smaller base of sales. These effects compound, leading to a greater overall loss of consumer welfare. In contrast, when the merging firms are similar in size, both price increases tend to be more moderate, producing less harm.

To build intuition, consider a simple example with single-product firms under logit demand. Suppose there are two mergers, one between firms with market shares $(s_A,s_B) = (0.2,0.2)$ and another between firms with shares $(s_A,s_B) = (0.8, 0.05)$. Both mergers generate the same change in concentration, $\Delta \mathit{HHI}_{AB} = 2s_As_B = 0.08$. Even though $\Delta \mathit{HHI}_{AB}$ is identical, the mergers differ in their implications for $\rho_1$.

Figure \ref{figure:rho.1.example} plots the combinations of $(s_A,s_B)$ that yield $\Delta \mathit{HHI}_{AB} = 0.08$. The solid line represents this set, while the dashed lines show the level curves of $\rho_1$ passing through the symmetric and asymmetric cases. In the symmetric case, $\rho_1 \approx 1.56$, whereas in the asymmetric case, $\rho_1 \approx 5.26$. Thus, the asymmetric merger produces a much higher value of $\rho_1$, indicating greater potential harm.

\begin{figure}[htbp]
    \centering
    \includegraphics[width=0.38\linewidth]{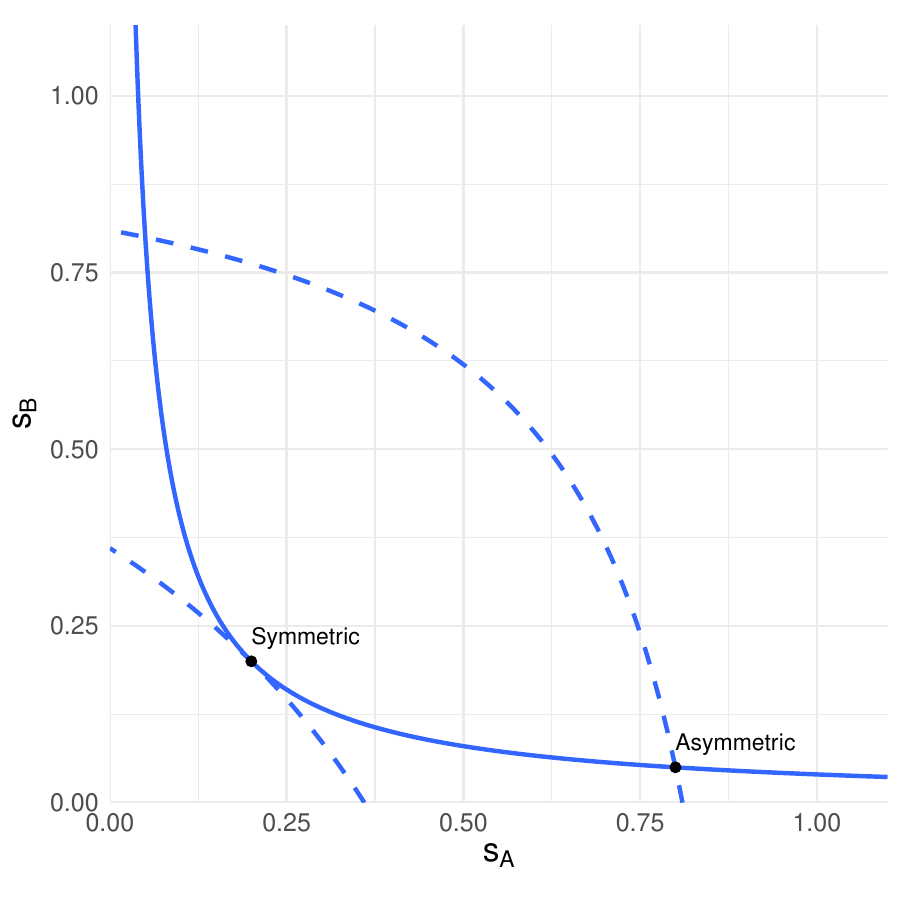}
    \caption{Example $\rho_1$ for Symmetric vs. Asymmetric Firm Shares }
    \floatfoot{\emph{Notes:} This figure illustrates that a merger between firms with asymmetric shares can generate a higher value of $\rho_1$ than a merger between firms with symmetric shares, even when both have the same $\Delta$HHI. The solid line represents the set of $(s_A,s_B)$ pairs satisfying $2s_As_B = 0.08$. The inner dashed line shows the level curve of $\frac{1}{(1-s_A)(1-s_B)}$ passing through $(s_A,s_B) = (0.2,0.2)$, while the outer dashed line shows the corresponding level curve passing through $(s_A,s_B) = (0.8, 0.05)$.}
    \label{figure:rho.1.example}
\end{figure}

This pattern arises because reducing one firm's share by a small amount allows for a much larger increase in the other firm's share to keep $\Delta \mathit{HHI}_{AB}$ constant. The trade-off is therefore not one-to-one, and as asymmetric increases, the combined shares $s_A + s_B$ can grow substantially even though $\Delta \mathit{HHI}_{AB}$ remains fixed.\footnote{A more precise statement would require comparing the curvatures of the level curve of $\Delta \mathit{HHI}_{AB}$ and that of $\rho_1$.}

To compare the consumer harm, assume that the upward pricing pressure closely approximates the actual price change ($\Delta p_j \approx \mathit{UPP}_j$), so that the effect of $\rho_2$ can be ignored. In the symmetric merger, $(\mathit{UPP}_A, \mathit{UPP}_B) = (0.3125/\alpha, 0.3125/\alpha)$. In the asymmetric merger, $(\mathit{UPP}_A, \mathit{UPP}_B) = (0.2632/\alpha, 4.2105/\alpha)$. In the asymmetric case, the larger firm applies a modest price increase to a substantial customer base, while the smaller firm applies a steep increase to a smaller base. The effects combine to generate a much larger total loss in consumer welfare. The change in consumer surplus is $\Delta \mathit{CS} = 0.125 (N/\alpha) $ in the symmetric merger and $\Delta \mathit{CS} = 0.4211 (N/\alpha) $ in the asymmetric merger.\footnote{If firm $A$ mergers with firm $B$, the upward pricing pressure for single-product firm $A$ facing logit demand with price sensitivity parameter $\alpha$ is $\mathit{UPP}_A = D_{A \to B} \times m_B = \frac{s_B}{1-s_A}\times \frac{1}{\alpha (1-s_B)}$. Approximating the price increase as $\Delta p_A \approx \mathit{UPP}_A$, the resulting change in consumer surplus is $\Delta \mathit{CS}_A \approx \mathit{UPP}_A \times q_A = \frac{N}{\alpha} \frac{s_As_B}{(1-s_A)(1-s_B)}$, where $q_A = s_A N$. $\Delta \mathit{CS}_B$ can be derived symmetrically.} Thus, even though both mergers have the same $\Delta \mathit{HHI}_{AB}$, the asymmetric merger is more than three times as harmful to consumers.

The idea that two mergers can generate the same $\Delta$HHI while yielding very different values of $\rho_1$ applies more broadly. Define the set of share combinations that yield a given change in concentration:
\begin{equation}\label{equation:shares.set}
    \mathcal{S}(c_0, \Delta_0) \equiv \{(s_A,s_B):\; 2s_As_B = \Delta_0, \; s_A + s_B \leq c_0 \},
\end{equation}
where the restriction $s_A + s_B \leq c_0$ prevents shares from becoming too arbitrary. For each pair of $(s_A,s_B)$ in this set, there is a corresponding value of $\rho_1$. Figure \ref{figure:rho.bounds} plots the upper and lower bounds of $\rho_1$ for $c_0 = 0.9$. The lower bound occurs when the merging firms are symmetric, while the upper bound occurs when shares are highly asymmetric.\footnote{Note that $\mathcal{S}(c_0, \Delta_0)$ is non-empty only if $c_0^2/2 \geq \Delta_0$. On an $(s_A,s_B)$ plane, the level curve of $2s_As_B = \Delta_0$ crosses the line $s_A + s_B = c_0$ only if $c_0^2/2 \geq \Delta_0$. It is straightforward to verify that $\rho_1^\mathit{MNL}$ is minimized when $s_A=s_B = \sqrt{\Delta_0/2}$ and maximized when $s_A = \frac{1}{2}(c_0 - \sqrt{c_0^2 - 2\Delta_0})$ and $s_B = \frac{1}{2} (c_0 + \sqrt{c_0^2 - 2\Delta_0})$. The lower and upper bounds are $\underline{\rho}_1^\mathit{MNL} = \frac{1}{(1 - \sqrt{\Delta_0 / 2})^2}$ and $\overline{\rho}_1^\mathit{MNL} = \frac{2}{\Delta_0 - 2c_0 + 2}$.} As the figure demonstrates, the cross-firm coefficient $\rho_1$ can vary substantially depending on how market shares are distributed between the merging firms, particularly when $\Delta$HHI is small.

\begin{figure}[htbp!]
    \centering
    \includegraphics[width = 0.5\textwidth]{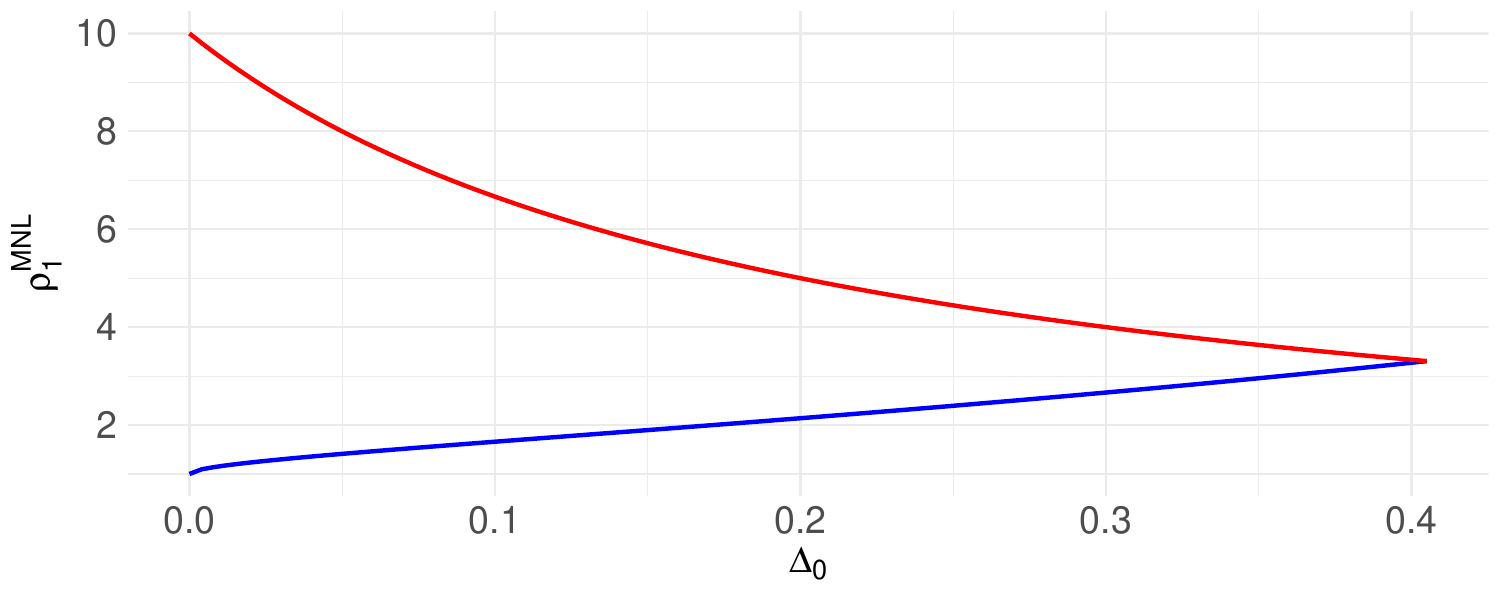}
    \caption{Upper and Lower Bound of $\rho_1^{MNL}$ for $c_0 = 0.9$}
    \label{figure:rho.bounds}
    \floatfoot{\emph{Notes}: This figure illustrates that mergers with the same $\Delta$HHI can differ substantially in their values of $\rho_1$. The upper contour represents the maximum value of $\rho_1$ for a given $\Delta \mathit{HHI} = \Delta_0$, which occurs when the merging firms' shares are most asymmetric. The lower contour represents the minimum value of $\rho_1$, attained when the merging firms' shares are symmetric.}
\end{figure}

\subsection{Within-Firm Scaling Factor}
The term $\rho_2$ in \eqref{equation:within.firm.scaling.factor} serves as a within-firm scaling factor, capturing how market share is distributed across the products of a multi-product firm. In principle, $\rho_2$ depends on both within-firm share distribution and cross-firm interactions through the merger pass-through matrix $M$. However, drawing on \citet{miller2017upward}'s observation that the diagonal elements of $M$ tend to dominate the off-diagonal elements, we can approximate the matrix by $M \approx \varphi I$ (as supported by Lemma \ref{lemma:small.share.approximation.of.merger.pass.through.matrix}). Under this simplification, $\rho_2$ simplifies to:
\[
\rho_2 \approx \frac{1}{2} \sum_{l \in \mathcal{J}_A} \frac{s_l / (\varphi - s_l)}{s_A / (\varphi - s_A)} + \frac{1}{2} \sum_{l \in \mathcal{J}_B} \frac{s_l / (\varphi - s_l)}{s_B / (\varphi - s_B)} . 
\]
Here, $\rho_2$ is the average of two additively separable components, each depending only on the distribution of product-level shares within a firm (holding the total firm share fixed). If each firm produces a single product, $\rho_2$ collapses to $\rho_2 \approx 1$. Thus, $\rho_2$ becomes relevant primarily when firms operate multiple products.\footnote{Without the approximation of the merger pass-through matrix, $\rho_2 = \frac{1}{2\varphi} (M_{AA} + M_{AB} \frac{s_A}{s_B} + M_{BA}\frac{s_B}{s_A} + M_{BB})$ (see Remark \ref{remark:single.product.case}). Thus, even for single-product firms, $\rho_2$ can depend on the merging firms' shares. However, when the shares are small, $M \approx \rho I$, making this dependence minor and supporting the interpretation of $\rho_2$ as primarily a scaling factor for multiproduct settings.}

The following result establishes that approximating $\rho_2$ as 1 is reasonable when merging firms have small market shares, regardless of the number of products.
\begin{proposition} \label{proposition:small.share.approximation.of.rho.2}
    $\rho_2 \to 1$ as the merging firms' market shares approach zero.
\end{proposition}
\begin{proof}
    See Appendix \ref{section:proof.of.proposition.3}.
\end{proof}
Proposition \ref{proposition:small.share.approximation.of.rho.2} implies that for mergers involving small firms, analysts can safely use the approximation $\rho_2 \approx 1$ without calculating the merger pass-through matrix.

Further analytical comparative statics for $\rho_2$ are difficult because it is a complex function of product- and firm-level shares, which also affect the pass-through matrix. To better understand its behavior, I rely on simulations under the logit demand model.\footnote{I find similar patterns under the CES demand assumption.} In these simulations, I set the price responsiveness parameter to $\alpha = 1$. Importantly, I do not impose the simplification $M \approx \varphi I$ when calculating $\rho_2$; instead, I compute the exact merger pass-through matrix.

I first examine the case where each firm produces a single product, even though $\rho_2$ is typically interpreted in a multi-product setting. This exercise clarifies the role of cross-firm share interactions that operate through the pass-through matrix. 

Figure \ref{figure:rho2.and.rho.under.asymmetry} uses contour plots to display how $\rho_2$ and the product $\rho_1 \times \rho_2$ vary with the firms' shares $(s_A, s_B)$. Figure \ref{figure:rho2.under.asymmetry} shows that, as predicted by Proposition \ref{proposition:small.share.approximation.of.rho.2}, $\rho_2 \approx 1$ when both firms' shares are small. However, $\rho_2$ is nonlinear and non-monotonic, with its value depending in a complex way on the relative shares of the two firms. Increasing asymmetry between firms' shares (holding $\Delta$HHI fixed) generally increases $\rho_2$, but the effect can be ambiguous in certain regions. Yet, Figure \ref{figure:rho.under.asymmetry} shows that, when considering the combined term $\rho_1 \times \rho_2$, increasing asymmetry unambiguously raises the value of $\rho_1 \times \rho_2$ for a fixed $\Delta$HHI.

\begin{figure}[htbp!]
\centering
\begin{subfigure}{.5\textwidth}
  \centering
  \includegraphics[width=.99\linewidth]{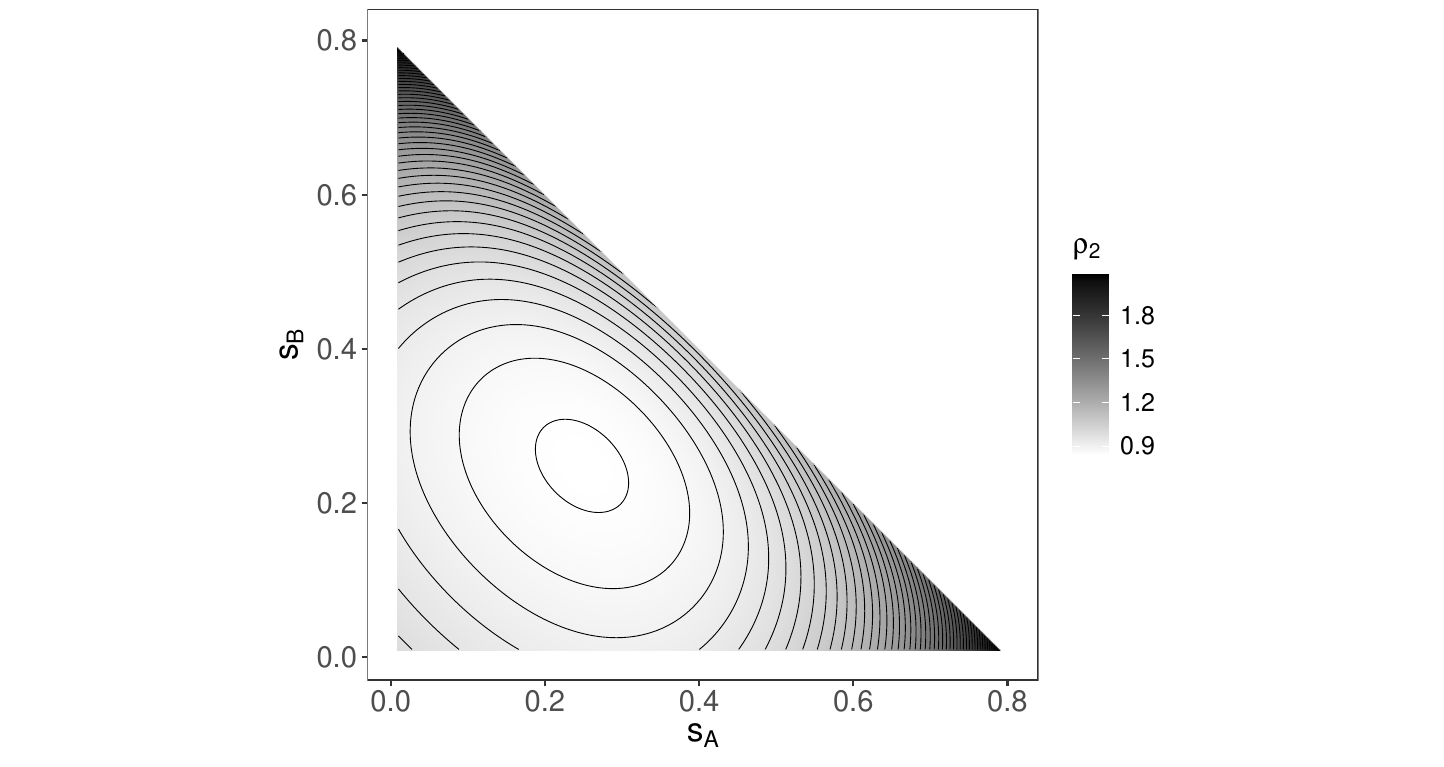}
  \caption{$\rho_2$}
  \label{figure:rho2.under.asymmetry}
\end{subfigure}%
\begin{subfigure}{.5\textwidth}
  \centering
  \includegraphics[width=.99\linewidth]{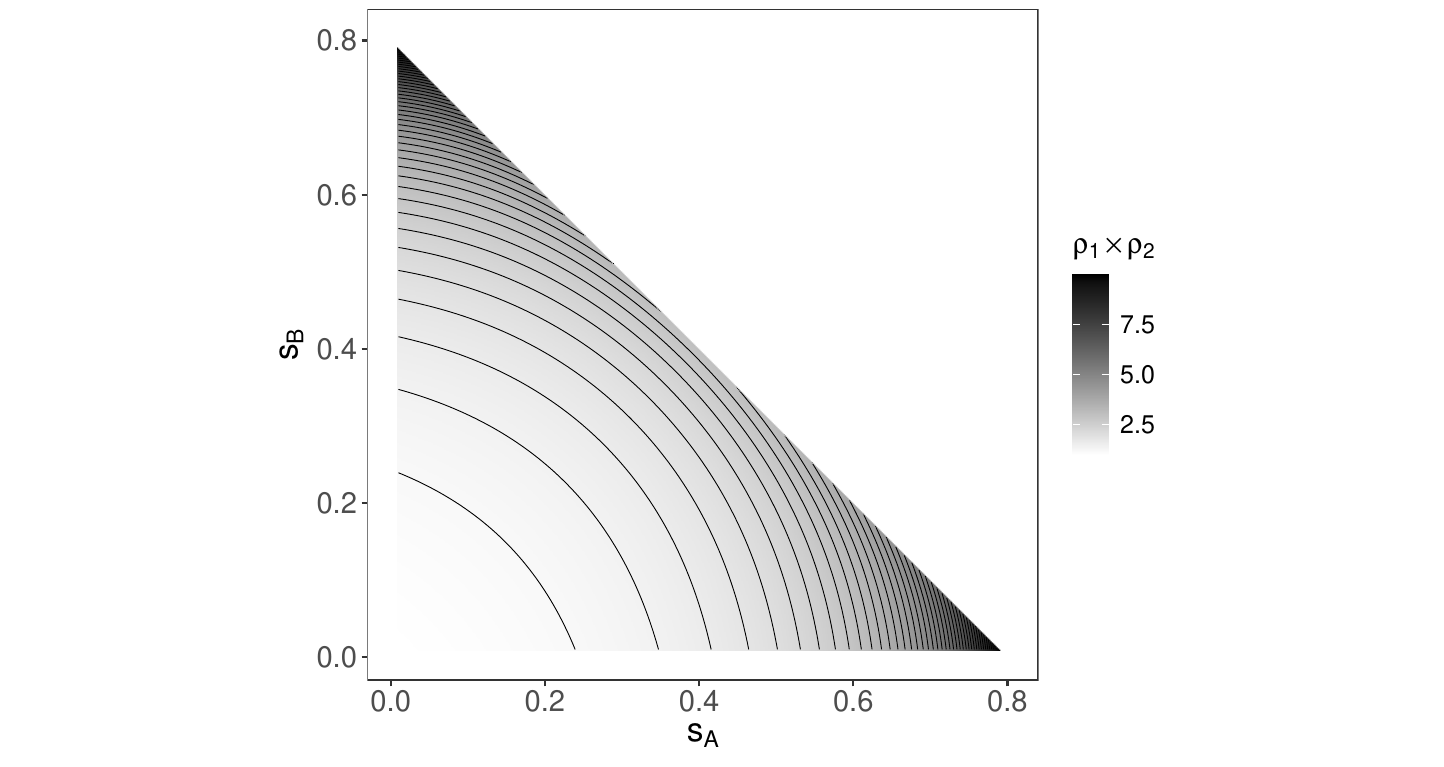}
  \caption{$\rho_1 \times \rho_2$}
  \label{figure:rho.under.asymmetry}
\end{subfigure}
\caption{The Values of $\rho_2$ and $\rho_1 \times \rho_2$ in Single-Product Logit Case}
\label{figure:rho2.and.rho.under.asymmetry}
\floatfoot{\emph{Notes}: The figure shows how the values of $\rho_2$ and $\rho_1 \times \rho_2$ vary with respect to $(s_A,s_B)$ under the single-product logit demand assumption. I constrain the set of $(s_A,s_B)$ to be those satisfying $s_A + s_B \leq 0.8$.}
\end{figure}

Next, I analyze how $\rho_2$ behaves when each firm produces multiple products with symmetric product-level shares. This exercise clarifies the role of total shares in multi-product firms and the number of products they produce. Suppose each firm has a total market share $s_f$ and produces $J$ products, so that each product has a share $s_j = s_f / J$. Figure \ref{figure:rho2.and.rho.under.symmetry} plots $\rho_2$ and $\rho_1 \times \rho_2$ as functions of $s_f$ and $J$. Figure \ref{figure:rho2.under.symmetry} remains non-monotonic in firm shares, consistent with the single-product results in Figure \ref{figure:rho2.under.asymmetry}. Increasing the number of products $J$ raises $\rho_2$ for a fixed $s_f$, but the effect is negligible unless the firms have very large market shares. On the other hand, Figure \ref{figure:rho.under.symmetry} shows that the combined term $\rho_1 \times \rho_2$ is monotonic in firm shares. Having more products has only a minor effect unless the firms are very large.

\begin{figure}[htbp!]
\centering
\begin{subfigure}{.5\textwidth}
  \centering
  \includegraphics[width=.99\linewidth]{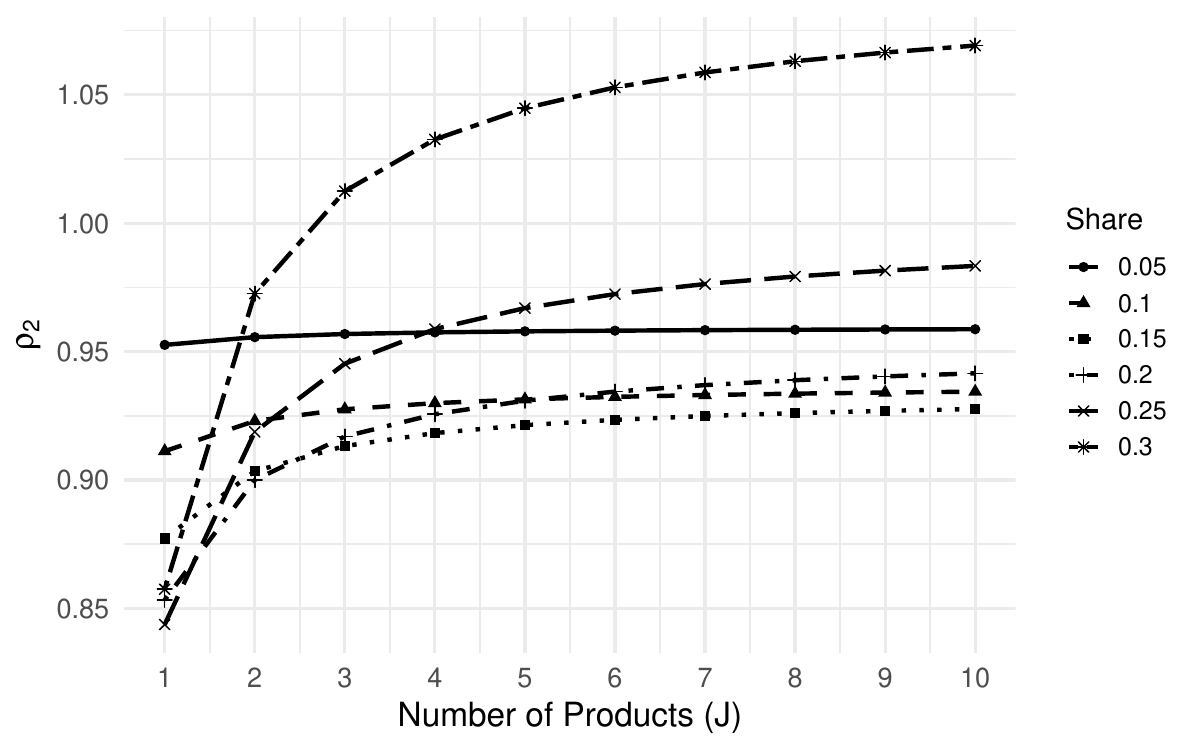}
  \caption{$\rho_2$}
  \label{figure:rho2.under.symmetry}
\end{subfigure}%
\begin{subfigure}{.5\textwidth}
  \centering
  \includegraphics[width=.99\linewidth]{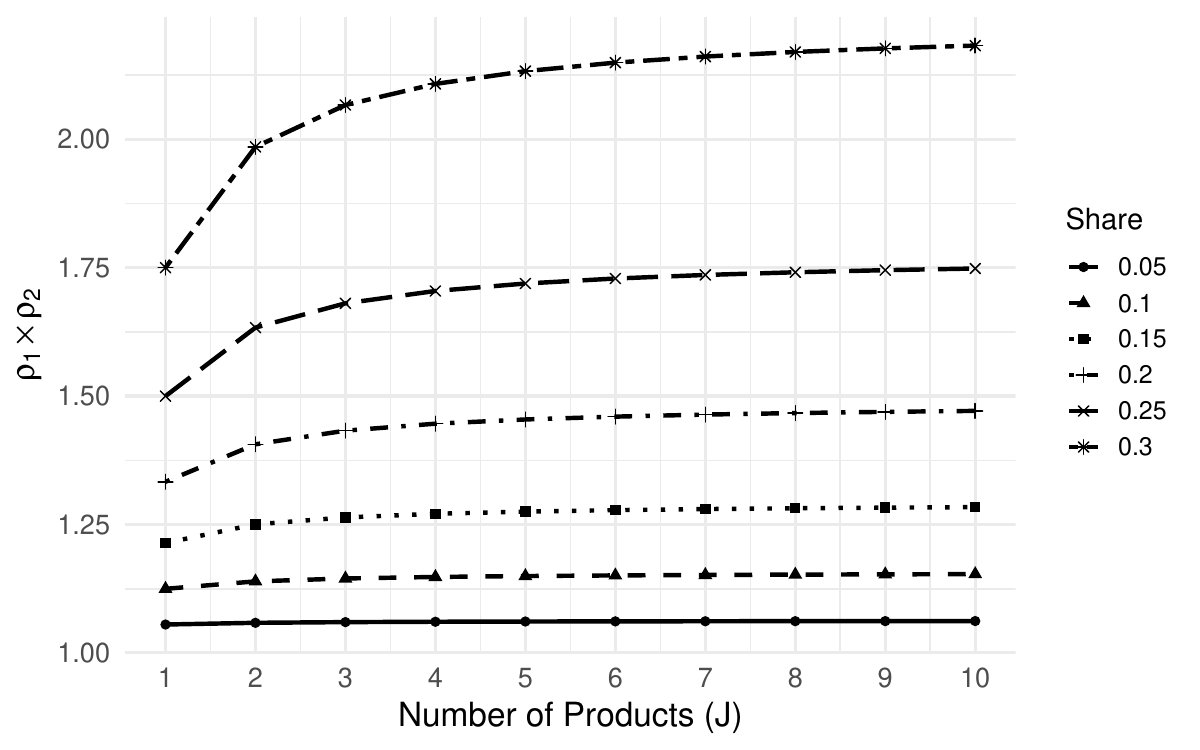}
  \caption{$\rho_1 \times \rho_2$}
  \label{figure:rho.under.symmetry}
\end{subfigure}
\caption{The Values of $\rho_2$ and $\rho_1 \times \rho_2$ for Symmetric Multiproduct Firms}
\label{figure:rho2.and.rho.under.symmetry}
\floatfoot{\emph{Notes}: This figure plots the values of $\rho_2$ and $\rho_1 \times \rho_2$ assuming symmetric firms that produce $J$ products with equal market shares under the logit demand assumption. ``Share'' indicates the value of each firm's share $s_f$.}
\end{figure}

The simulations show that $\rho_2$ is typically close to one when merging firms have small or moderate market shares, supporting the practical approximation $\rho_2 = 1$. While $\rho_2$ can exhibit nonlinear and non-monotonic behavior---particularly for large firms or highly uneven product mixes---these effects are generally modest and only become significant when firm shares are very large. Moreover, its influence is dominated by $\rho_1$, as the combined term $\rho_1 \times \rho_2$ rises predictably with firm size and asymmetry, holding $\Delta$HHI fixed. Thus, in most merger analyses, $\rho_2$ acts as a minor adjustment factor, and the key variation in outcomes is driven primarily by $\rho_1$.
\section{Comparison to \citet{nocke2023aggregative} \label{section:comparison.to.Nocke.Shutz}}
\subsection{Comparison of Formulas}
Equation \eqref{equation:delta.cs.and.delta.hhi} in Proposition \ref{proposition:consumer.surplus.first.order.approach} aligns with the findings of \citet{nocke2022concentration} and \citet{nocke2023aggregative}: consumer harm from a merger is proportional to $\Delta\mathit{HHI}$. Specifically, \citet{nocke2023aggregative} Proposition 8 shows that when market shares are small, the market power effect of a merger without synergies can be approximated as
\begin{equation}\label{equation:nocke.shutz.2023.result}
    \Delta \mathit{CS} \approx 
    \begin{cases}
        - V_0^\mathit{MNL} \Delta \mathit{HHI}_{AB} & \text{(MNL)} \\
        - V_0^\mathit{CES} \frac{\sigma - 1}{\sigma} \Delta \mathit{HHI}^R_{AB} & \text{(CES)}
    \end{cases}
    ,
\end{equation}
where the proportionality coefficient is constant and does not depend on the merging firms' individual market shares.

The key difference between their result in \eqref{equation:nocke.shutz.2023.result} and my expression in \eqref{equation:delta.cs.and.delta.hhi} lies in the underlying approximation method. \citet{nocke2023aggregative} relies on a second-order Taylor expansion around the limit of small market shares and monopolistic competition. By contrast, equation \eqref{equation:delta.cs.and.delta.hhi} is derived using the first-order approach of \citet{farrell2010antitrust} and \citet{jaffe2013first}. This distinction matters: Proposition \ref{proposition:consumer.surplus.first.order.approach} shows that, in my framework, pre-merger market share distributions directly affect merger harm through the scaling factors $\rho_1$ and $\rho_2$. In other words, even conditional on $\Delta \mathit{HHI}$, the way market shares are distributed---both across and within the merging firms---remain economically relevant.

Importantly, the two formulations are fully consistent in the limit of small market shares. From Propositions \ref{proposition:cross.firm.scaling.factor.1} and \ref{proposition:small.share.approximation.of.rho.2}, as $(s_A, s_B) \to (0,0)$, we have $\rho_1 \to 1/ \varphi$ and $\rho_2 \to 1$. Substituting these limits into equation \eqref{equation:delta.cs.and.delta.hhi} yields $\Delta \mathit{CS} \to \frac{V_0}{\rho} \mathit{HHI}_{AB}$, which exactly matches the characterization in equation \eqref{equation:nocke.shutz.2023.result}. Thus, my result is consistent with \citet{nocke2023aggregative} but provides a richer picture of how share distribution shapes merger effects when firms have non-negligible market shares.

\subsection{Monte Carlo Experiment}
I run Monte Carlo experiments to examine how my formula \eqref{equation:delta.cs.and.delta.hhi} performs relative to \citet{nocke2023aggregative}'s formula \eqref{equation:nocke.shutz.2023.result}. I follow the simulation design in \citet{miller2017upward}.\footnote{I consider a market with six single-product firms. I randomly draw the share and margin of firm 1 and use these to calibrate the model parameters. I then simulate a merger between firms 1 and 2 and compute the resulting change in consumer surplus, normalizing the market size so that $V_0 = 1$. Appendix \ref{section:simulation} provides further details on the Monte Carlo experiments.} I apply my approximation $M \approx \varphi I$ (Lemma \ref{lemma:small.share.approximation.of.merger.pass.through.matrix}) to gauge how they perform without requiring the analyst to calculate pass-through rates, which are generally complex functions of model parameters and data.\footnote{In Appendix \ref{section:misspecification}, I also examine the role of model misspecification. I find that using the MNL formula generally predicts a larger consumer welfare loss from a merger than the CES formula.}

Monte Carlo experiments indicate that Proposition \ref{proposition:consumer.surplus.first.order.approach} improves upon the approximation in \citet{nocke2023aggregative}. Figure \ref{figure:consumer.surplus.simulation.comparison} compares the two approaches: their formula in \eqref{equation:nocke.shutz.2023.result} and my first-order approach in \eqref{equation:delta.cs.and.delta.hhi}. For the logit case, the Nocke-Schutz formula tends to underestimate the welfare effects (Figure \ref{figure:consumer.surplus.simulation.logit.1}). My approach provides a closer approximation on average, although it can exhibit greater variance and occasionally overestimate welfare effects when the actual change in consumer surplus is large (Figure \ref{figure:consumer.surplus.simulation.logit.2}). This difference arises because \eqref{equation:nocke.shutz.2023.result} omits the $\rho_1^\mathit{MNL}$ term, which is always greater than one. The CES results are similar. Again, the Nocke-Schutz formula underestimates the true effects (Figure \ref{figure:consumer.surplus.simulation.ces.1}). In contrast, the first-order approach performs quite well, with a dense clustering of points along the 45-degree line (Figure \ref{figure:consumer.surplus.simulation.ces.2})

\begin{figure}[htbp!]
\centering
\begin{subfigure}{.5\textwidth}
  \centering
  \includegraphics[width=.5\linewidth]{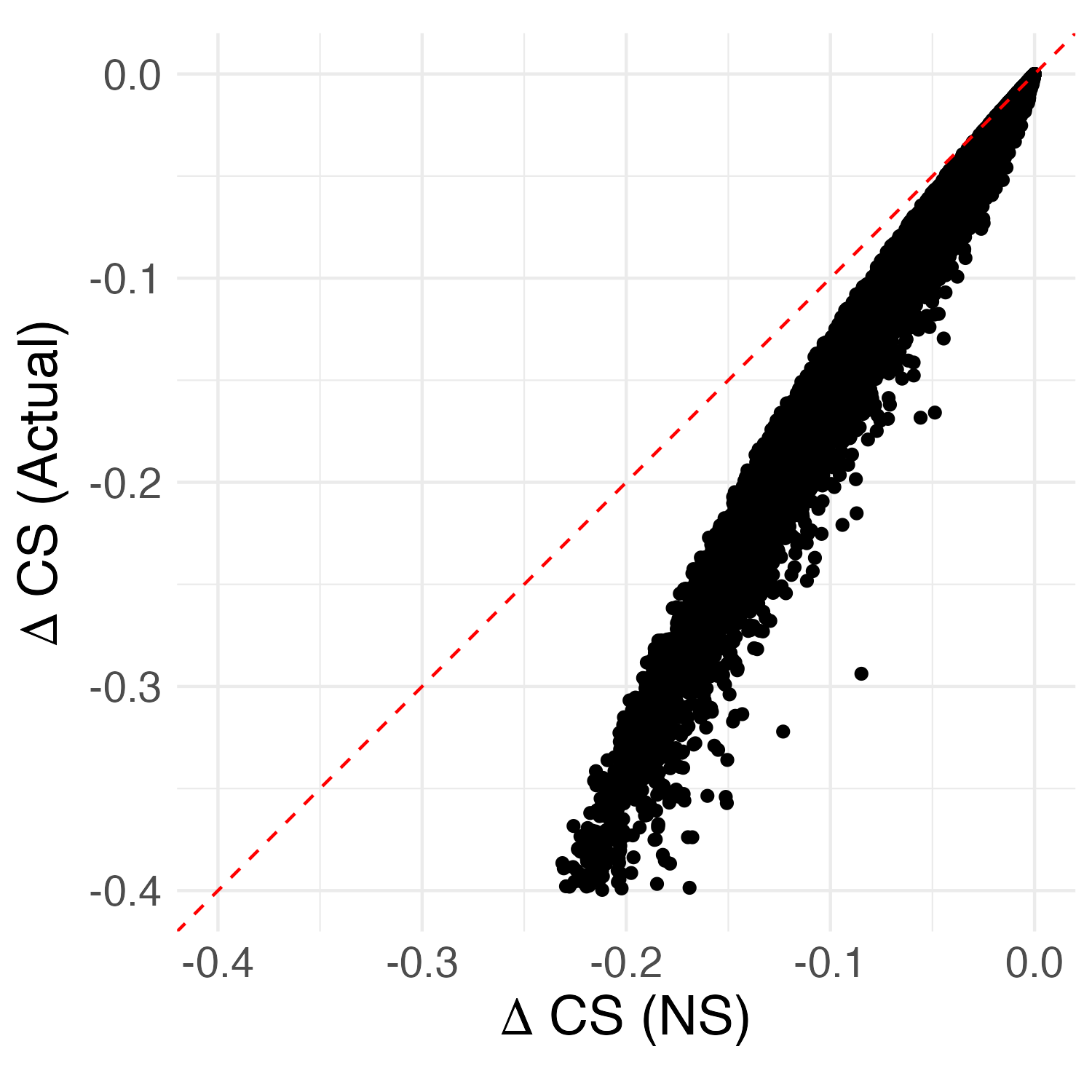}
  \caption{Nocke-Schutz (Logit)}
  \label{figure:consumer.surplus.simulation.logit.1}
\end{subfigure}%
\begin{subfigure}{.5\textwidth}
  \centering
  \includegraphics[width=.5\linewidth]{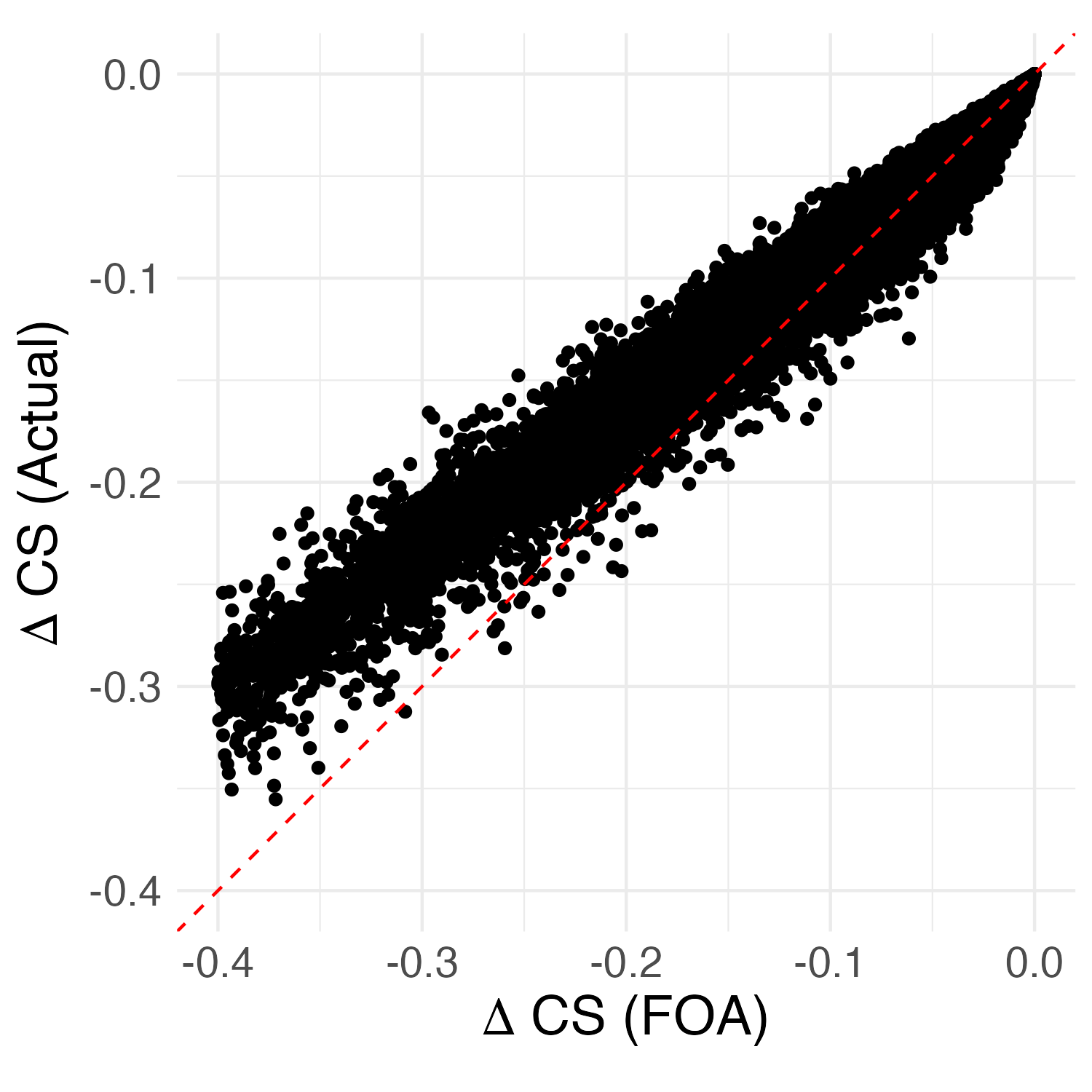}
  \caption{First-Order Approach (Logit)}
  \label{figure:consumer.surplus.simulation.logit.2}
\end{subfigure}
\vskip \baselineskip
\begin{subfigure}{.5\textwidth}
  \centering
  \includegraphics[width=.5\linewidth]{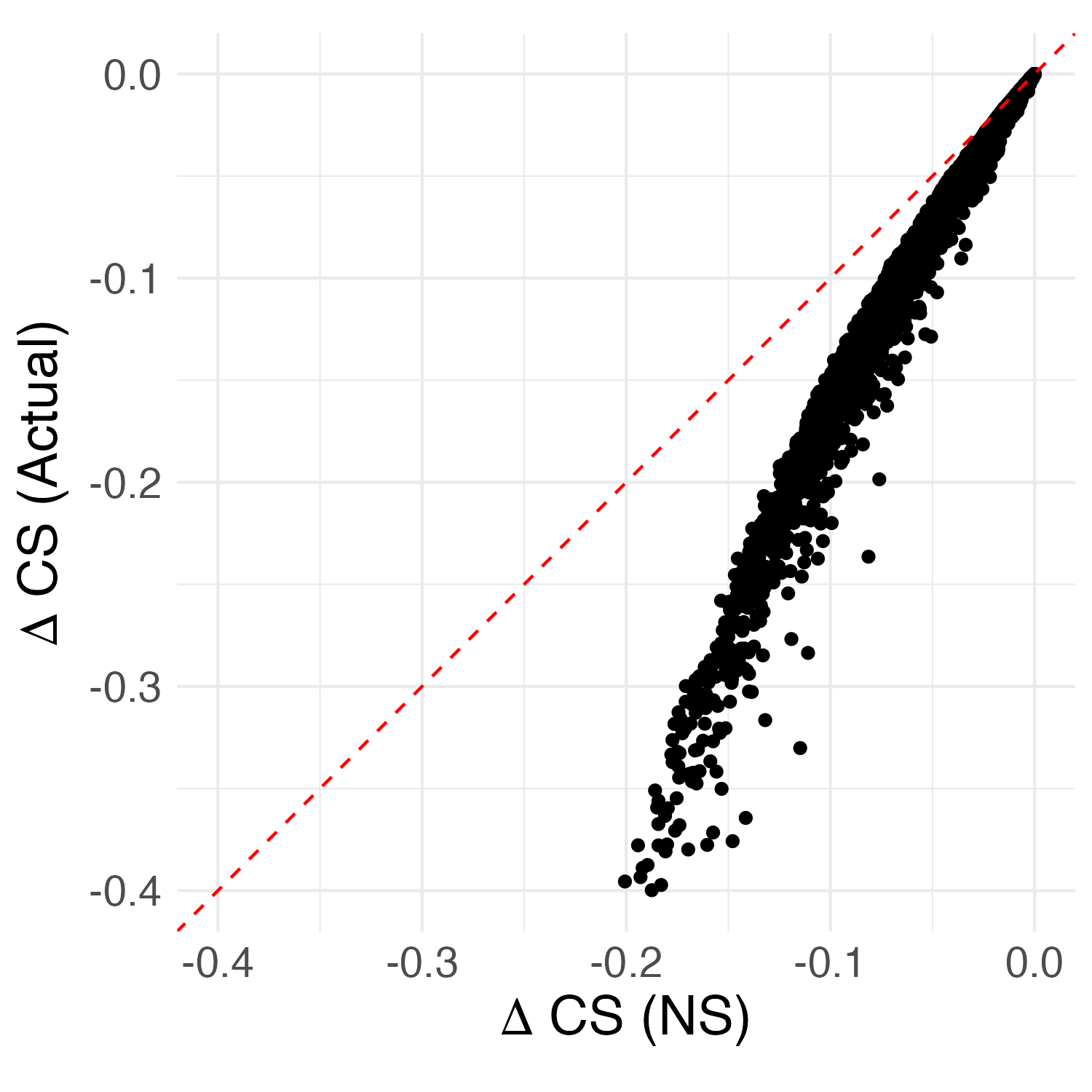}
  \caption{Nocke-Schutz (CES)}
  \label{figure:consumer.surplus.simulation.ces.1}
\end{subfigure}%
\begin{subfigure}{.5\textwidth}
  \centering
  \includegraphics[width=.5\linewidth]{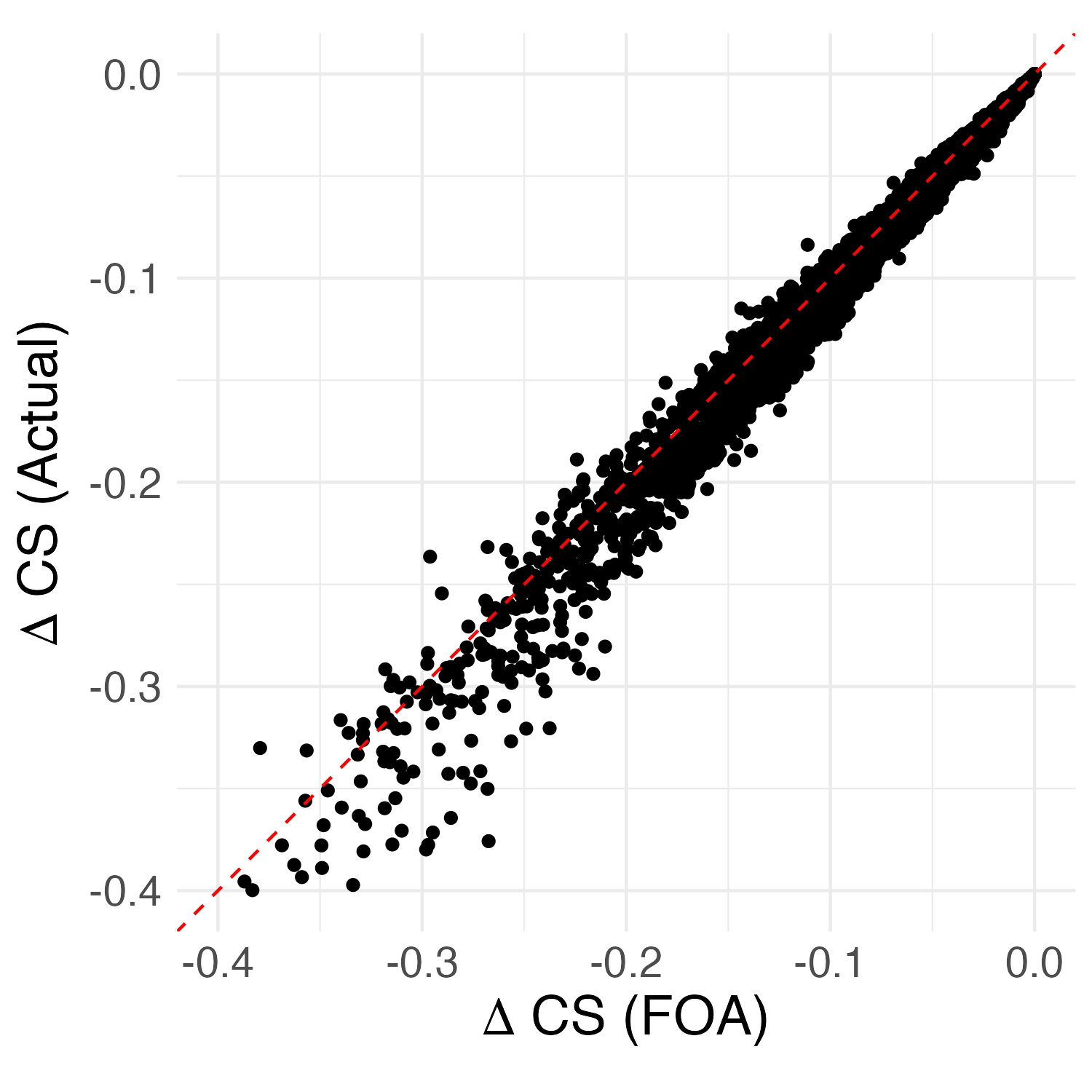}
  \caption{First-Order Approach (CES)}
  \label{figure:consumer.surplus.simulation.ces.2}
\end{subfigure}
\caption{Comparison of Nocke-Schutz and First-Order Approach}
\label{figure:consumer.surplus.simulation.comparison}
\floatfoot{\emph{Notes}: This figure compares the performance of the Nocke-Schutz formula $\Delta \mathit{CS}^\mathit{NS} = - (V_0 / \varphi) \Delta \mathit{HHI}_{AB}$ and my first-order approach-based formula $\Delta \mathit{CS}^\mathit{FOA} = - V_0 \rho_1 \rho_2 \Delta \mathit{HHI}_{AB}$ after setting $\rho_2 = 1$ in predicting the actual consumer welfare loss due to mergers.}
\end{figure}

Figure \ref{figure:consumer.surplus.percentage.deviation} reports the distribution of prediction error, defined as the percentage deviation between the predicted and actual change in consumer welfare.\footnote{Specifically, for each simulation, I calculate the prediction error as $\log (\Delta \mathit{CS}^\text{predicted} / \Delta \mathit{CS}^\text{actual})$.} For the logit (CES) case, the average errors are -36\% (-28\%) for the Nocke-Schutz approach and -4\% (-10\%) for the first-order approach, respectively. Overall, the first-order approach substantially improves the accuracy of linking $\Delta \mathit{HHI}$ to predicted changes in consumer welfare from mergers, offering a much closer fit than the Nocke-Schutz approximation.

\begin{figure}[htbp!]
\centering
\begin{subfigure}{.5\textwidth}
  \centering
  \includegraphics[width=.99\linewidth]{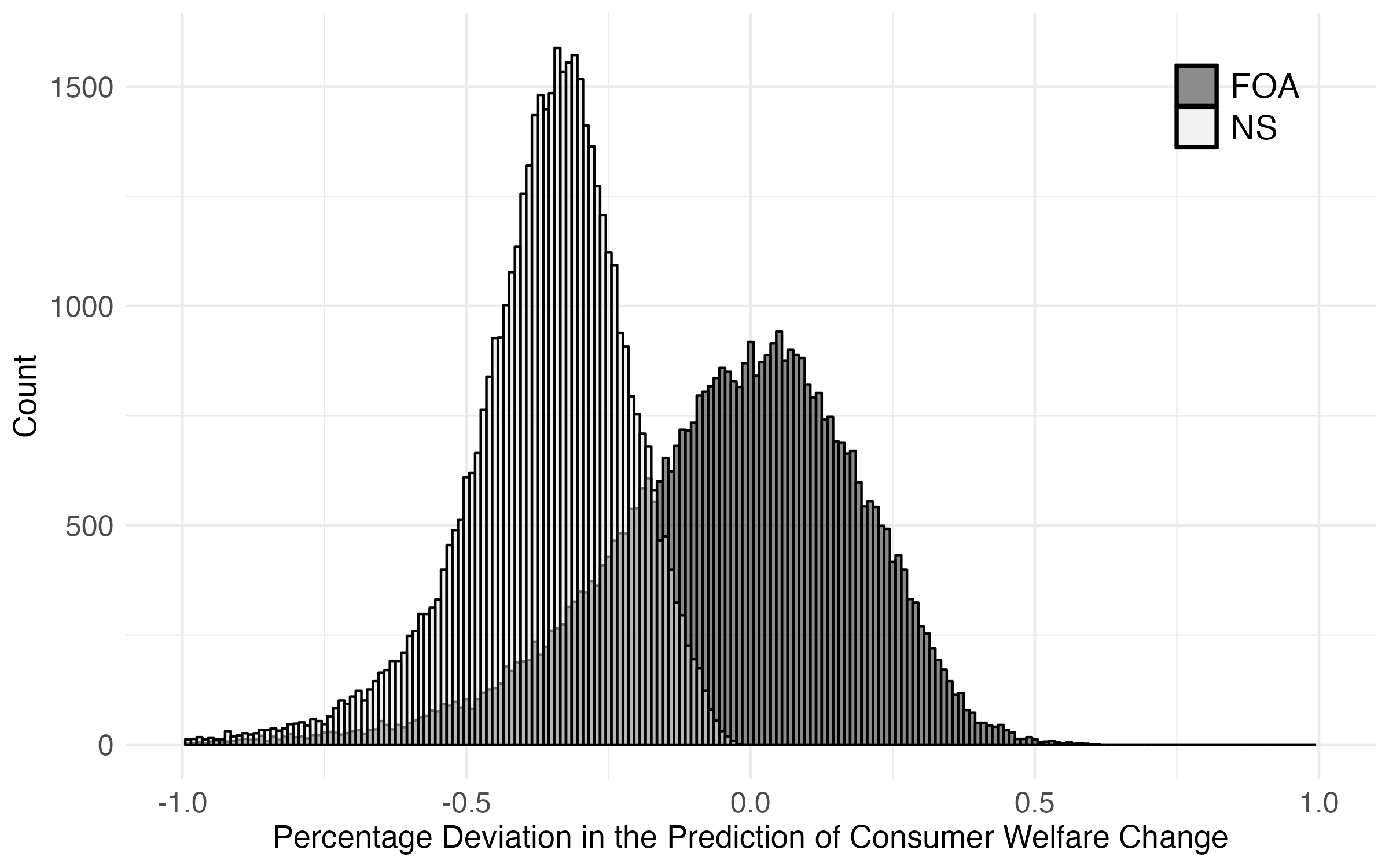}
  \caption{Logit}
  \label{figure:consumer.surplus.percentage.deviation.logit}
\end{subfigure}%
\begin{subfigure}{.5\textwidth}
  \centering
  \includegraphics[width=.99\linewidth]{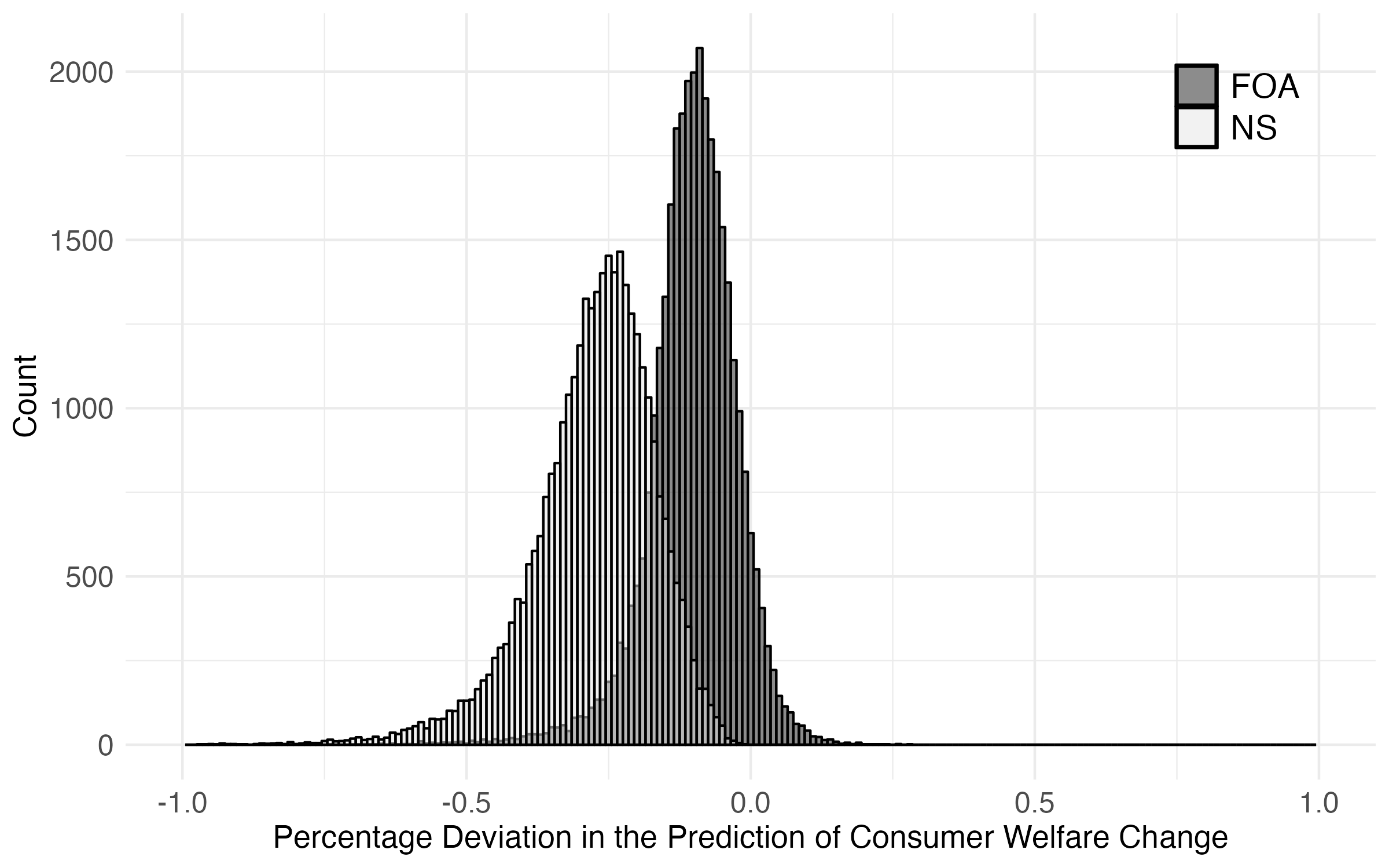}
  \caption{CES}
  \label{figure:consumer.surplus.percentage.deviation.logit}
\end{subfigure}
\caption{Distribution of Prediction Errors in the Nocke-Schutz and First-Order Approach}
\label{figure:consumer.surplus.percentage.deviation}
\floatfoot{\emph{Notes}: This figure plots the distribution of prediction errors of the Nocke-Schutz formula and my first-order approach-based formula. For each simulation draw, the prediction error is measured as $\log(\Delta \mathit{CS}^\text{predicted}/\Delta \mathit{CS}^\text{actual})$. }
\end{figure}
\section{An Empirical Example: \emph{Heinz/Beech-Nut} Merger \label{section:empirical.example}}

To demonstrate how my framework can be applied in practice, I consider the proposed merger between Heinz and Beech-Nut in the U.S. baby food market.\footnote{I choose this example because it aligns well with the theoretical framework of the paper, and firm-level market share data is readily available.} In 2000, H.J. Heinz announced plans to acquire Beech-Nut for \$185 million. At the time, the U.S. baby food market was valued between \$865 million and \$1 billion annually and was dominated by three firms: Gerber with a 65\% market share, Heinz with 17.4\%, and Beech-Nut with 15.4\% (measured by revenue). Had the merger proceeded, the industry would have effectively become a duopoly, with the combined Heinz-Beech-Nut entity controlling 98\% of the market \citep{chen2010evolution}.

The Federal Trade Commission challenged the merger under Section 7 of the Clayton Act, arguing that reducing the number of major competitors from three to two would significantly harm competition. Heinz defended the deal by claiming substantial merger-specific efficiencies. In particular, it argued that closing Beech-Nut's outdated plant and streamlining production and distribution would lower costs, enabling the merged firm to compete more effectively against Gerber. However, the D.C. Circuit ultimately rejected these arguments, ruling that the claimed efficiencies were not sufficient to offset the likely competitive harm. In response, Heinz abandoned the merger in 2001.

For tractability, I model Heinz and Beech-Nut as single-product firms facing constant elasticity of substitution (CES) demand and assume there are no merger-specific efficiencies.\footnote{Product-level share data are not publicly available.} To be conservative, I set the total market size to $Y = \$865$ million, corresponding to the lower bound of the estimated industry size. Based on the pre-merger market share, the change in concentration implied by the merger is $\Delta \mathit{HHI} = 0.0536$. Using this value, I compute estimates of consumer harm by $\Delta \mathit{CS} = -\rho \Delta \mathit{HHI}$, where $\rho = V_0 \rho_1 \rho_2$ is the proportionality coefficient for translating concentration changes to consumer surplus loss. I estimate $\rho_2$ along with the CES merger pass-through matrix without relying on $M \approx \varphi I$.

Table \ref{table:heinz.beech.nut} reports these estimates for different values of the CES price responsiveness parameter, $\sigma$.\footnote{Information on firms' margin, which is not publicly available, can produce an estimate of $\sigma$.} The results indicate that consumer harm would range from \$19.10 million to \$36.68 million annually, depending on the assumed level of consumer price sensitivity. As expected, harm is greater when consumers are less price sensitive.

\begin{table}[htbp!]
    \centering
    \caption{Annual Consumer Harm Calculation for the \emph{Heinz/Beech-Nut} Merger} \label{table:heinz.beech.nut}
    \begin{threeparttable}
    \begin{tabular}{ccccccccc} \toprule
        $\sigma$ & $\varphi$ & $1/\varphi$ &  $V_0$ & $\rho_1$ & $ \rho_2$ & $\rho$ & $\Delta \mathit{HHI}$ & $\Delta \mathit{CS}$ \\ \midrule
        1.5 & 3.00 & 0.33 & 1730.00 & 0.37 & 1.09 & 703.16 & 0.0536 & -\$37.68m \\
        2.0 & 2.00 & 0.50 & 865.00 & 0.59 & 1.04 & 531.37 & 0.0536 & -\$28.48m \\ 
        2.5 & 1.67 & 0.60 & 576.67 & 0.73 & 1.00 & 426.72 & 0.0536 & -\$22.87m \\
        3.0 & 1.50 & 0.67 & 432.00 & 0.84 & 0.98 & 356.39 & 0.0536 & -\$19.10m \\ \bottomrule
    \end{tabular}
    \begin{tablenotes}
        \item \footnotesize \emph{Notes}: $\rho_2$ is calculated with the pass-through matrix. Note that $\rho = V_0 \rho_1 \rho_2$ and $\Delta \mathit{CS} = -\rho \Delta \mathit{HHI}$.
    \end{tablenotes}
    \end{threeparttable}
\end{table}

Recall that as market shares approach zero, $\rho_1 \to 1/\varphi $ and $\rho_2 \to 1$. Table \ref{table:heinz.beech.nut} shows that while the approximation $\rho_2 \approx 1$ is generally reasonable, the approximation $\rho_1 \approx 1/\varphi$ can be quite inaccurate. This suggests that if an analyst wishes to adopt a simple, fast approach without computing the full merger pass-through matrix, setting $\rho_2 = 1$ is a defensible simplification; this finding is consistent with the simulation results reported in Section \eqref{section:comparison.to.Nocke.Shutz}.

\section{Conclusion \label{section:conclusion}}

This paper develops simple formulas linking changes in HHI to the consumer welfare effects of mergers in differentiated product markets with multi-product firms. Using a first-order approach, I show that merger-induced changes in consumer surplus are proportional to changes in HHI, with the proportionality coefficient determined by market size, the price responsiveness parameter, and the distribution of the merging firms’ shares. These results bridge concentration-based merger screening tools, commonly used by antitrust agencies, with the unilateral effects framework rooted in the standard Bertrand-Nash pricing model, providing a clearer theoretical foundation for interpreting changes in market concentration.

\part*{Appendix}
\appendix
\section{Proofs \label{section:proofs}}

\subsection{Proof of Lemma \ref{lemma:small.share.approximation.of.merger.pass.through.matrix}} \label{section:proof.of.lemma.1}

In the logit case, the merger pass-through matrix is defined as $M = - \left( \frac{\partial h(p^\text{pre})}{\partial p} \right)^{-1}$ where $h(p)$ is a function that characterizes the post-merger FOC such that $h(p^\text{post}) = 0$ but normalized to be quasilinear in marginal cost. By expressing each $\frac{\partial h_j(p)}{\partial p_k}$ with $j,k \in \mathcal{J}_A \cup \mathcal{J}_B$ as a function of quantity-based shares and the price sensitivity parameter $\alpha$, it can be verified that $\frac{\partial h_j}{\partial p_j} \to -1$ as the shares converge to zero, and $\frac{\partial h_j}{\partial p_k} \to 0$ for any pair $j \neq k$.\footnote{I characterize the merger pass-through matrix under the logit and CES demand assumptions in Online Appendix \ref{section:merger.pass.through.matrix.under.logit.and.ces}. } Thus, $\frac{\partial h}{\partial p} \to -I$, which gives $M = I$.
    
In the CES case, the merger pass-through matrix is defined as $M = - \left( \frac{\partial h(\tilde{p^\text{pre}})}{\partial \tilde{p}} \right)^{-1}$, where $\tilde{p}_j \equiv \log p_j$ and $h(\tilde{p})$ is a function that characterizes the post-merger FOC such that $h(\tilde{p}^\text{post}) = 0$ but normalized to be quasilinear in relative margins $m_j \equiv \frac{p_j - c_j}{p_j}$. Similar to the logit case, one can express each $\frac{\partial h_j(\tilde{p})}{\partial \tilde{p}_k}$ with $j , k \in \mathcal{J}_A \cup \mathcal{J}_B$ as functions of revenue-based shares and the price sensitivity parameter $\sigma$. It can be verified that $\frac{\partial h_j}{\partial \tilde{p}_j} \to - \frac{\sigma - 1}{\sigma}$ as the shares converge to zero, and $\frac{\partial h_j}{\partial \tilde{p}_k} \to 0$ for any pair $j \neq k$. Thus, $\frac{\partial h}{\partial \tilde{p}} \to - \left( \frac{\sigma -1}{\sigma} \right)I$, which gives $M \to \frac{\sigma}{\sigma - 1} I$. \qed

\subsection{Proof of Proposition \ref{proposition:consumer.surplus.first.order.approach} }\label{section:proof.of.proposition.1}

Before proceeding to the proof, I show how absolute (resp. relative) margins relate to firms' shares in the logit (resp. CES) case. In \citet{nocke2018multiproduct}, the authors define the ``$\iota$-markups'' as
\begin{equation}\label{equation:iota.markups}
    \mu_f \equiv 
    \begin{cases}
    \alpha (p_j - c_j),  & \forall j \in \mathcal{J}_f \quad \text{(MNL)} \\
    \sigma \frac{p_j - c_j}{p_j},  & \forall j \in \mathcal{J}_f \quad \text{(CES)}
    \end{cases}
    .
\end{equation}
The $\iota$-markups use the fact that with logit (resp. CES) demand, firms charge constant absolute (relative) margins across their products. \citet{nocke2023aggregative} equation (5) shows that
\begin{equation}\label{equation:iota.markup.and.firm.share}
    \mu_f = \frac{1}{1 - \alpha^* s_f}
\end{equation}
holds in equilibrium, where $\alpha^* = 1$ under logit and $\alpha^* = \frac{\sigma - 1}{\sigma}$ under CES.\footnote{Also see the derivation in the Online Appendix A of \citet{caradonna2024mergers}.} Combining \eqref{equation:iota.markups} and \eqref{equation:iota.markup.and.firm.share} implies 
\begin{equation}\label{equation:margin.in.logit}
    p_j - c_j = \frac{1}{\alpha (1-s_f^Q)}, \quad \forall j \in \mathcal{J}_f
\end{equation}
in the case of logit, and
\begin{equation}\label{equation:margin.in.ces}
    \frac{p_j - c_j}{p_j} = \frac{1}{1 + (1-s_f^R)(\sigma-1)}, \quad \forall j\in \mathcal{J}_f
\end{equation}
in the case of CES.

\subsubsection{Logit}

    Under logit demand, quantity diversion ratio is $D_{j \to k} = \frac{s_k^Q}{1-s_j^Q}$ and firms' absolute margins satisfy \eqref{equation:margin.in.logit}. Then I can write the upward pricing pressure of $j \in \mathcal{J}_A$ as 
    \[
    \mathit{UPP}_j = \sum_{k \in \mathcal{J}_B} (p_k - c_k) D_{j \to k} = \sum_{k \in \mathcal{J}_B} \frac{1}{\alpha (1-s_B^Q)} \frac{s_k^Q}{1-s_j^Q} = \frac{s_B^Q}{\alpha(1-s_B^Q)(1-s_j^Q)}.
    \]
    Thus, for an arbitrary product $j \in \mathcal{J}_A \cup \mathcal{J}_B$, I can express its upward pricing pressure as
    \[
    \mathit{UPP}_j = \mathbb{I}_{j \in \mathcal{J}_A} \frac{s_B^Q}{\alpha (1-s_B^Q)(1-s_j^Q)} + \mathbb{I}_{j \in \mathcal{J}_B} \frac{s_A^Q}{\alpha(1-s_A^Q)(1-s_j^Q)}.
    \]

    Next, given that the first order price effect on an arbitrary product $j$ is $\Delta p_j = \sum_{k \in \mathcal{J}_A \cup \mathcal{J}_B} M_{jk} \mathit{UPP}_k$, we have
    \begin{align*}
        \Delta \mathit{CS}_j & = -\Delta p_j \times q_j \\
        & = - \left(\sum_{k \in \mathcal{J}_A \cup \mathcal{J}_B} M_{jk} \mathit{UPP}_k \right)s_j^QN \\
        & = - \left( \sum_{k \in \mathcal{J}_A \cup \mathcal{J}_B } M_{jk} \left( \mathbb{I}_{k \in \mathcal{J}_A} \frac{s_B^Q}{\alpha (1-s_B^Q)(1-s_k^Q)} + \mathbb{I}_{k \in \mathcal{J}_B} \frac{s_A^Q}{\alpha (1-s_A^Q)(1-s_k^Q)} \right) \right) s_jN \\
        & = - \left( \frac{N}{\alpha} \right) \left( \frac{s_A^Q s_B^Q}{(1-s_A^Q)(1-s_B^Q)} \right) \sum_{k \in \mathcal{J}_A \cup \mathcal{J}_B} M_{jk} \frac{s_j^Q}{s_k^Q} \left( \mathbb{I}_{k \in \mathcal{J}_A} \frac{s_k^Q / (1-s_k^Q)}{s_A^Q/(1-s_A^Q)} + \mathbb{I}_{k \in \mathcal{J}_B} \frac{s_k^Q/(1-s_k^Q)}{s_B^Q/(1-s_B^Q)} \right).
    \end{align*}
    Finally, using $\Delta \mathit{CS} = \sum_{j \in \mathcal{J}_A \cup \mathcal{J}_B} \Delta \mathit{CS}_j$ gives
    \begin{align*}
        \Delta \mathit{CS} & = - \left( \frac{N}{\alpha} \right) \left( \frac{s_A^Q s_B^Q}{(1-s_A^Q)(1-s_B^Q)} \right) \sum_{j,k \in \mathcal{J}_A \cup \mathcal{J}_B} M_{jk} \frac{s_j^Q}{s_k^Q} \left( \mathbb{I}_{k \in \mathcal{J}_A} \frac{s_k^Q / (1-s_k^Q)}{s_A^Q/(1-s_A^Q)} + \mathbb{I}_{k \in \mathcal{J}_B} \frac{s_k^Q/(1-s_k^Q)}{s_B^Q/(1-s_B^Q)} \right) \\
        & = - \left( \frac{N}{\alpha} \right) \left( \frac{\Delta \mathit{HHI}^Q_{AB}}{(1-s_A^Q)(1-s_B^Q)} \right) \sum_{j,k \in \mathcal{J}_A \cup \mathcal{J}_B} M_{jk} \frac{s_j^Q}{s_k^Q} \left( \frac{1}{2} \mathbb{I}_{k \in \mathcal{J}_A} \frac{s_k^Q / (1-s_k^Q)}{s_A^Q/(1-s_A^Q)} + \frac{1}{2} \mathbb{I}_{k \in \mathcal{J}_B} \frac{s_k^Q/(1-s_k^Q)}{s_B^Q/(1-s_B^Q)} \right).
    \end{align*}
\qed

\subsubsection{CES}

Under a CES preference assumption, upward pricing pressure has a close relationship with \emph{revenue diversion ratios}, defined as $D_{j \to l}^R \equiv - (\partial R_l / \partial p_j)/(\partial R_j / \partial p_j)$.  \citet{koh2024merger} shows that for $j \neq k$,
\begin{equation}\label{equation:relationship.between.quantity.and.revenue.diversion.ratios}
(1 + \epsilon_{jj}^{-1}) D_{j \to k}^R = D_{j \to k} \frac{p_k}{p_j},    
\end{equation}
where $\epsilon_{jj} = (\partial q_j / \partial p_j) / (q_j / p_j)$ is product $j$'s own-price elasticity of demand. Plugging in \eqref{equation:relationship.between.quantity.and.revenue.diversion.ratios} into the GUPPI ($\mathit{GUPPI}_j \equiv \mathit{UPP}_j / p_j$) equation gives
\begin{equation}
    \mathit{GUPPI}_j = (1 + \epsilon_{jj}^{-1}) \sum_{l \in \mathcal{J}_B} m_l D_{j \to l}^R.
\end{equation}

I can simplify the $\mathit{GUPPI}_j$ expression as follows. Under CES demand, the revenue diversion ratio is $D_{j \to l}^R = \frac{s_l^R}{1-s_j^R}$ \citep{koh2024merger, caradonna2024mergers} and the relative margins satisfy \eqref{equation:margin.in.ces}. Using the fact that
\[
R_j = \frac{v_jp_j^{1-\sigma}}{1 + \sum_{l \in \mathcal{J}} v_l p_l^{1-\sigma}} Y,
\]
I can also show that the own-price elasticity of revenue is $\epsilon_{jj}^R = \frac{\partial s_j^R}{\partial p_j} \frac{p_j}{s_j^R} = -(1-s_j^R)(\sigma - 1).$ Then, since $\epsilon_{jj}^R = \epsilon_{jj} + 1$ \citep{koh2024merger}, I have
\[
1 + \epsilon_{jj}^{-1} = \frac{\epsilon_{jj}^R}{\epsilon_{jj}^R - 1} = \frac{(1-s_j^R)(\sigma - 1)}{1 + (1 - s_j^R)(\sigma - 1)}.
\]
Then, $\mathit{GUPPI}_j$ for $j \in \mathcal{J}_A$ simplifies to
\begin{align*}
    \mathit{GUPPI}_j & = (1 + \epsilon_{jj}^{-1})\sum_{k \in \mathcal{J}_B} m_k D_{j \to k}^R \\
    & = \left( \frac{(1-s_j^R)(\sigma - 1)}{1 + (1-s_j^R)(\sigma - 1)} \right) \sum_{k \in \mathcal{J}_B} \frac{1}{1 + (1-s_B^R)(\sigma - 1)} \frac{s_k^R}{1-s_j^R} \\ 
    & = \frac{(\sigma - 1) s_B^R}{(1 + (1-s_j^R)(\sigma - 1))(1 + (1-s_B^R)(\sigma - 1))}.
\end{align*}
In sum, $\mathit{GUPPI}_j$ for an arbitrary $j \in \mathcal{J}_A \cup \mathcal{J}_B$ can be rewritten as
\begin{align*}
    \mathit{GUPPI}_j & = \mathbb{I}_{j \in \mathcal{J}_A} \frac{(\sigma - 1) s_B^R}{(1 + (1-s_j^R)(\sigma-1))(1+(1-s_B^R)(\sigma-1))}  \\ &+ \mathbb{I}_{j \in \mathcal{J}_B} \frac{(\sigma-1)s_A^R}{(1 + (1-s_j^R)(\sigma - 1))(1 + (1-s_A^R)(\sigma-1))}.
\end{align*}

Next, letting $R_j = p_jq_j$,
\begin{align*}
    \Delta \mathit{CS}_j &= - \Delta p_j \times q_j \\
    & = - \frac{\Delta p_j}{p_j} \times R_j \\
    & = - \left(\sum_{k \in \mathcal{J}_A \cup \mathcal{J}_B} M_{jk} \mathit{GUPPI}_k \right) s_j^R Y \\
    & = - \left( \frac{Y}{\sigma -1} \right) \left( \frac{ \Delta \mathit{HHI}_{AB}^R}{(\frac{\sigma}{\sigma -1} - s_A^R)(\frac{\sigma}{\sigma -1} - s_B^R) } \right) \times \\ & \sum_{k \in \mathcal{J}_A \cup \mathcal{J}_B} M_{jk} \frac{s_j^R}{s_k^R} \left( \frac{1}{2} \mathbb{I}_{k \in \mathcal{J}_A} \frac{s_k^R/(1 + (1-s_k^R)(\sigma -1))}{s_A^R/(1 + (1-s_A^R)(\sigma-1))} + \frac{1}{2} \mathbb{I}_{k \in \mathcal{J}_B} \frac{s_k^R/(1 + (1-s_k^R)(\sigma -1))}{s_B^R/(1 + (1-s_B^R)(\sigma-1))} \right)
\end{align*}
Finally, since $\Delta \mathit{CS} = \sum_{j \in \mathcal{J}_{A} \cup \mathcal{J}_B} \Delta \mathit{CS}_j$, summing over the above expression over $j \in \mathcal{J}_A \cup \mathcal{J}_B$ gives the desired expression. \qed

\subsection{Proof of Proposition \ref{proposition:cross.firm.scaling.factor.1} \label{section:proof.of.proposition.2}}
First, that $\rho_1 = \frac{\varphi}{(\varphi -s_A)(\varphi - s_B)} \to \frac{1}{\varphi}$ as $(s_A,s_B) \to (0,0)$ is straightforward from the definition of $\rho_1$. Second, the positive monotonicity of $\rho_1$ with respect to the merging firms is trivial. Finally, verifying that $\rho_1$ is strictly convex in $(s_A,s_B)$ if $s_A < \varphi$ and $s_B < \varphi$ is also straightforward and can be checked by examining the positive definiteness of the Hessian, which is 
\[
H = 
\begin{pmatrix}
    \frac{2\varphi}{(\varphi - s_A)^3(\varphi - s_B)} & \frac{\varphi}{(\varphi - s_A)^2(\varphi - s_B)^2} \\ \frac{\varphi}{(\varphi - s_A)^2(\varphi - s_B)^2} & \frac{2 \varphi}{(\varphi - s_A)(\varphi - s_B)^3}
\end{pmatrix}
.
\]
\qed

\subsection{Proof of Proposition \ref{proposition:small.share.approximation.of.rho.2}} \label{section:proof.of.proposition.3}

To capture the assumption that the market shares are approaching zero, I assume merging firms' product-level shares are approaching zero at the same rate. Let $s_j = \gamma_j s_f$, where $\gamma_j > 0$ is a constant such that $\sum_{j \in \mathcal{J}_f} \gamma_j = 1$. Thus, I can consider the behavior of $\rho_2$ as $s_f \to 0$ for $f = A,B$. I also assume that the two firms' shares are approaching zero at the same rate. Then, for each $l \in \mathcal{J}_A \cup \mathcal{J}_B$, I have
\[
\frac{\frac{s_l}{1-s_l}}{\frac{s_f}{1-s_f}}  = \frac{\frac{\gamma_l s_f}{1-\gamma_l s_f}}{\frac{s_f}{1-s_f}} = \frac{\gamma_l (1-s_f)}{1 - \gamma_l s_f} \overset{s_f \downarrow 0}{\longrightarrow} \gamma_l.
\]
Recalling from Lemma \ref{lemma:small.share.approximation.of.merger.pass.through.matrix} that $M \to \varphi I$ as the shares approach zero, I have
\[
\begin{split}
    \rho_2 & = \frac{1}{\varphi} \sum_{j, l \in \mathcal{J}_A \cup \mathcal{J}_B} M_{jl} \frac{s_j}{s_l} \left( \frac{1}{2} \mathbb{I}_{l \in \mathcal{J}_A} \frac{s_l/(\varphi - s_l)}{s_A/(\varphi - s_A)} + \frac{1}{2} \mathbb{I}_{l \in \mathcal{J}_B} \frac{s_l/(\varphi - s_l)}{s_B / (\varphi - s_B)} \right) \\
    & \to \frac{1}{\varphi} \sum_{l \in \mathcal{J}_A \cup \mathcal{J}_B} \varphi \left( \frac{1}{2} \mathbb{I}_{l \in \mathcal{J}_A } \gamma_l + \frac{1}{2} \mathbb{I}_{l \in \mathcal{J}_B} \gamma_l \right) = \sum_{l \in \mathcal{J}_A \cup \mathcal{J}_B} \frac{\gamma_l}{2}  = 1,
\end{split}
\]
where the last equality follows from having $\sum_{l \in \mathcal{J}_A}\gamma_l = 1$ and $\sum_{l \in \mathcal{J}_B} \gamma_l = 1$. \qed

\section{Simulation \label{section:simulation}}
I follow the simulation steps in \citet{miller2017upward} with a slight modification. The overall steps are described as follows.
\begin{enumerate}
    \item Set $k=6$ as the number of firms. Assume all pre-merger prices are equal to one.\footnote{Setting all pre-merger prices to one allows us to equate the absolute change and the relative change in price.} Draw a random length-$(k+1)$ vector of shares from the Dirichlet distribution with parameter $\theta = (1,...,1) \in \mathbb{R}^7$. The shares represent quantity shares in the case of MNL and revenue shares in the case of CES. Note that outside option is included. Draw relative margin $m_1$ from the uniform distribution with support $[0.3, 0.6]$. I assume all firms' pre-merger prices are equal to 1.

    \item Using the observed shares and margin of firm 1, identify (calibrate) the price responsiveness parameter ($\alpha$ for MNL and $\sigma$ for CES) and the pre-merger values of $\{\mu_f\}_{f \in \mathcal{F}}$, $\{T_f\}_{f \in \mathcal{F}}$, and $H$ as follows. (See \citet{caradonna2024mergers} for derivations.) 
    \begin{itemize}
        \item In the MNL case, I use
        \begin{align}
            H & \leftarrow \frac{1}{1 - \sum_{f \in \mathcal{F}} s_f} \\
            \mu_f & \leftarrow \frac{1}{1 - s_f}, \quad \forall f \in \mathcal{F} \\
            T_f & \leftarrow H s_f \exp\left( \frac{1}{1-s_f} \right), \quad \forall f \in \mathcal{F} \\
            \alpha & \leftarrow \mu_1 / (p_1 - c_1) = \mu_1 / m_1
        \end{align}
        \item In the CES case, I use
        \begin{align}
            H & \leftarrow \frac{1}{1 - \sum_{f \in \mathcal{F}} s_f^R} \\
            \sigma & \leftarrow \left( \frac{1}{m_1} - 1 \right) \left( \frac{1}{1-s_1^R} \right) + 1 \\
            \mu_f & \leftarrow \frac{1}{1 - \frac{\sigma -1 }{\sigma} s_f^R}, \quad \forall f \in \mathcal{F} \\
            T_f & \leftarrow s_f^R H \left(1 - \frac{\mu_f}{\sigma} \right)^{1-\sigma}, \quad \forall f \in \mathcal{F}
        \end{align}
    \end{itemize}
    I also compute the upward pricing pressure using pre-merger values of margins and shares.

    \item I assume Firms 1 and 2 merge without synergy. Let $T_M \equiv T_1 + T_2$ be the type of the merged firm. I solve for the post-merger equilibrium using the system of equations described in Online Appendix \ref{section:review.of.aggregative.games.framework}.

    \item Using the pre-merger and post-merger equilibrium values, I calculate the actual price increases and changes in consumer surplus. The change in consumer surplus is calculated after normalizing $V_0 = 1$ in both logit and CES cases.

    \item I repeat the above $50,000$ times. I discard observations that fail to find the post-merger equilibrium due to numerical issues.\footnote{Although using multiple starting points can improve the chance of convergence, I instead chose to increase the number of simulation experiments.} 
\end{enumerate}

\section{Misspecification \label{section:misspecification}}

Considering demand function misspecification is crucial in merger analysis because inaccurate assumptions about how consumers substitute between products can bias estimates of diversion ratios, pass-through rates, and price elasticities. This not only distorts predictions of post-merger prices but also misstates changes in consumer surplus, which depend on the full shape of the demand curve. As a result, regulators may mistakenly approve harmful mergers, block beneficial ones, or design ineffective remedies, underscoring the importance of testing alternative demand models to ensure credible and robust welfare assessments.

I briefly discuss the implications of model misspecification in the context of the simulation exercise presented in Section \ref{section:comparison.to.Nocke.Shutz}. The simulation setting is convenient because pre-merger prices are normalized to one, allowing me to treat quantity shares and revenue shares as equivalent. Recall that under the single-agent assumption and the approximation $M \approx \varphi I$, the consumer welfare effect of a merger can be expressed as $\Delta \mathit{CS} \approx - V_0 \times r(\varphi) \times \Delta \mathit{HHI}_{AB}$, where $r(\varphi) \equiv \frac{\varphi}{(\varphi - s_A)(\varphi - s_B)}$. Note that $r(\varphi)$ is strictly decreasing in $\varphi$, as taking the derivative gives $r'(\varphi) = \frac{s_A s_B - \varphi^2}{(\varphi - s_A)^2(\varphi -s_B)^2} < 0$.

Since $\varphi = 1$ under the multinomial logit model, $\varphi = \frac{\sigma}{\sigma - 1} > 1$ under the CES model, and $r(\varphi)$ decreases with $\varphi$, it follows that $\vert \Delta \mathit{CS}^\mathit{MNL} \vert > \vert \Delta \mathit{CS}^\mathit{CES} \vert $. In other words, applying the MNL formula yields a larger predicted merger harm than using the CES formula. This result aligns with the earlier finding that the first-order approach with CES demand tends to underpredict merger price and welfare effects relative to MNL.

Figure \ref{figure:misspecification} illustrates this relationship by plotting predicted $\Delta \mathit{CS}^\mathit{MNL}$ against $\Delta \mathit{CS}^\mathit{CES}$ (with $V_0 = 1$). Overall, the MNL formula produces systematically higher welfare loss predictions. While the relative rankings of mergers are often preserved across the two models, they are not always identical. Since the gap between the two approximations is zero when $\varphi = 1$ and widens as $\varphi$ increases, we can infer that whether the rankings remain consistent depends on the estimated value of CES price responsiveness parameter $\sigma$.

\begin{figure}[htbp!]
\centering
\begin{subfigure}{.5\textwidth}
  \centering
  \includegraphics[width=.5\linewidth]{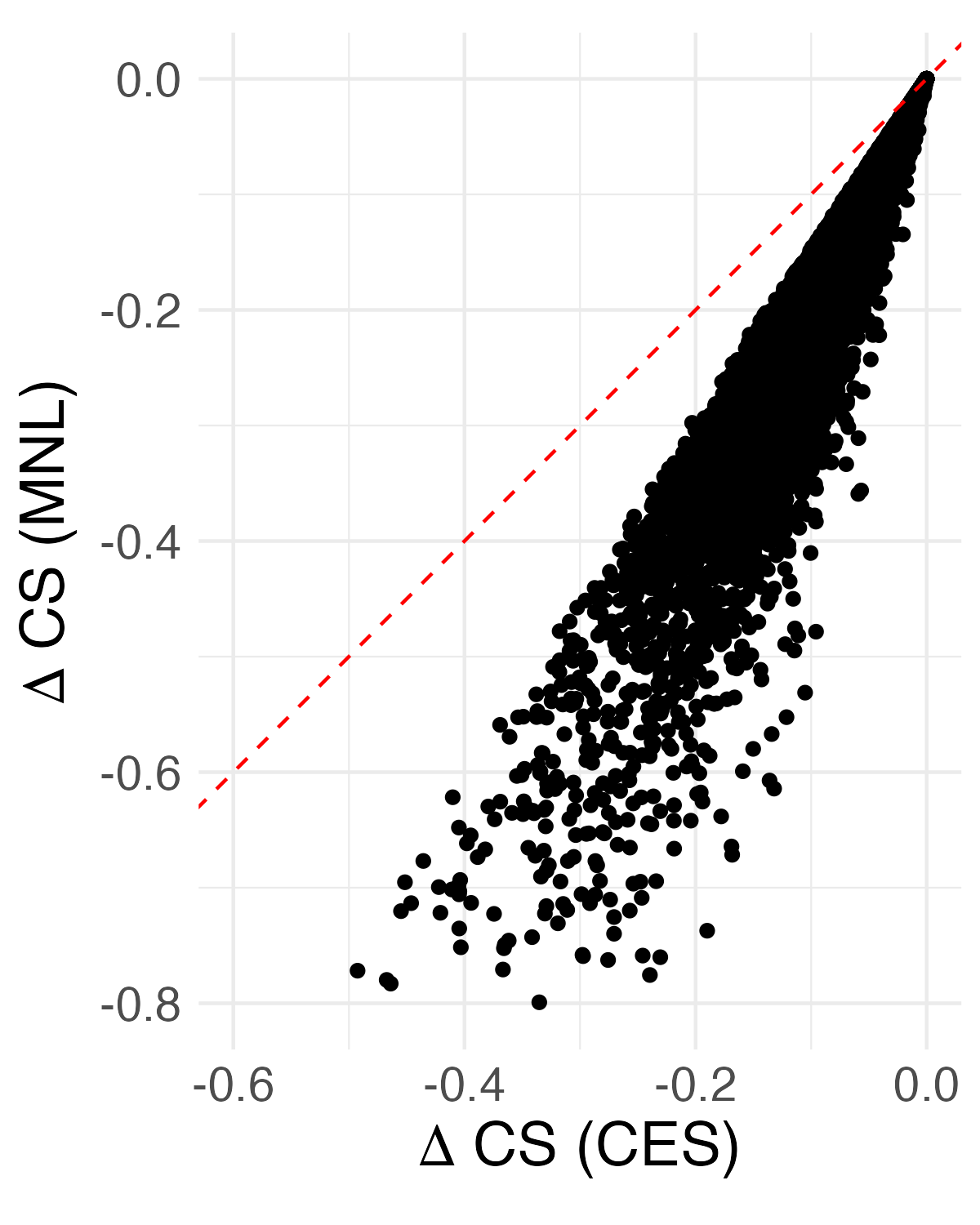}
  \caption{When DGP is Logit}
  \label{figure:misspecification.under.logit}
\end{subfigure}%
\begin{subfigure}{.5\textwidth}
  \centering
  \includegraphics[width=.5\linewidth]{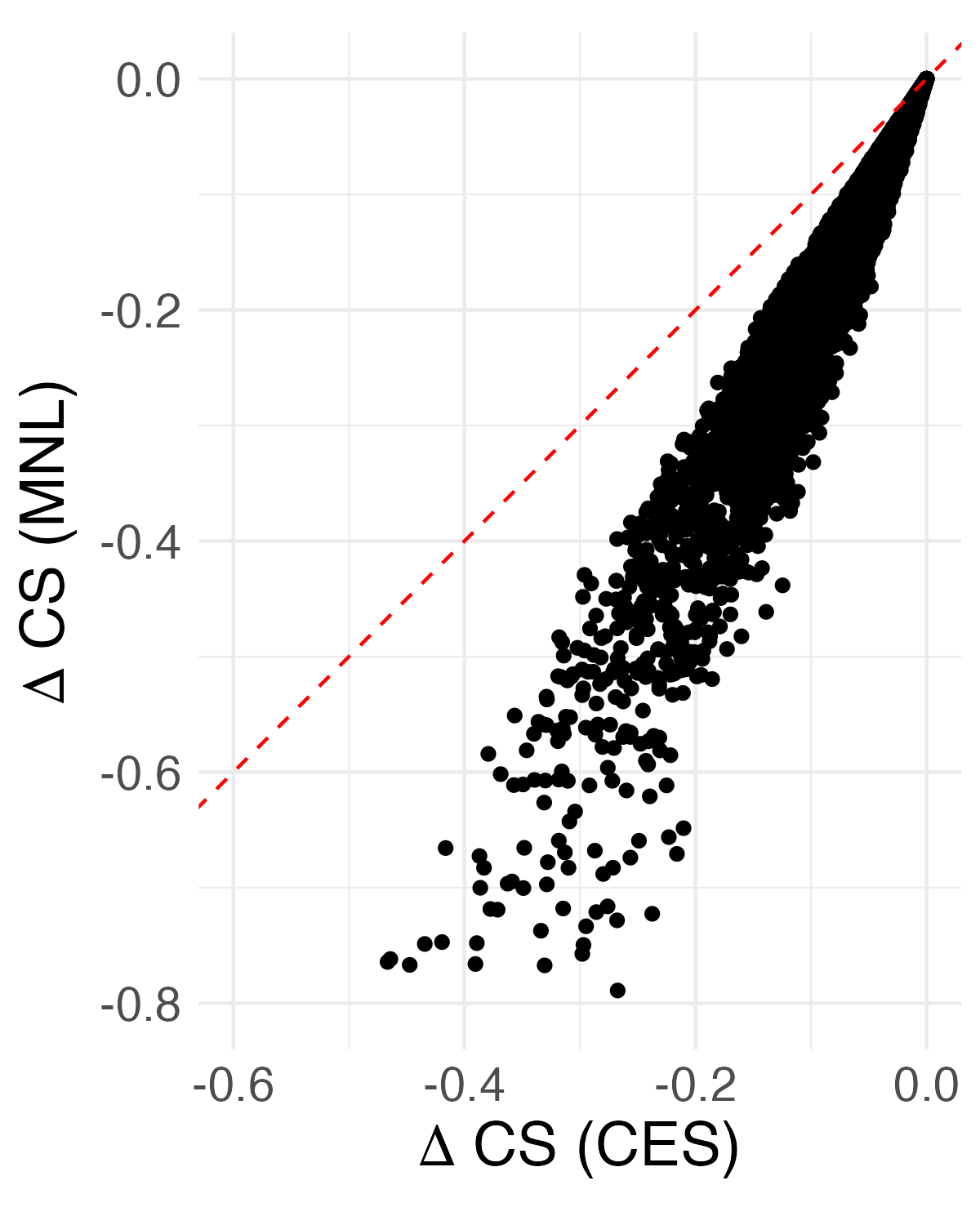}
  \caption{When DGP is CES}
  \label{figure:misspecification.under.ces}
\end{subfigure}
\caption{The Impact of Demand Misspecification on Prediction of Merger Welfare Effects}
\label{figure:misspecification}
\end{figure}

\singlespacing
\bibliographystyle{econ}
\bibliography{references}
\doublespacing

\clearpage

\setcounter{page}{1}  
\renewcommand{\thepage}{A\arabic{page}}  
\setcounter{section}{0} 

\begin{center}
    \LARGE Online Appendix for \\
    \LARGE{Concentration-Based Inference for Evaluating Horizontal Mergers} \\
    \vspace{1cm}
    \normalsize Paul S. Koh \\
    \normalsize September 20, 2025
\end{center}

\section{Review of Aggregative Games Framework \label{section:review.of.aggregative.games.framework}}

I review \citet{nocke2018multiproduct}'s aggregative games framework, which I use as a framework for the simulation exercises described in the paper.

\subsection{Model Primitives}
I consider \citet{nocke2018multiproduct}'s aggregative games framework with CES/MNL demand.\footnote{I also refer the readers to the appendix of \citet{caradonna2024mergers} for a practitioner's guide to the aggregation games framework of \citet{nocke2018multiproduct}.} A multiproduct-firm oligopoly model is specified as a tuple 
\begin{equation}\label{equation:model.primitives}
    \langle \mathcal{J}, \mathcal{F}, (c_j)_{j \in \mathcal{J}}, y, V_0, H_0, (h_j)_{j \in \mathcal{J}} \rangle.
\end{equation}
$\mathcal{J}$ is a finite set of products; $\mathcal{F}$ is a finite set of firms that form a partition over $\mathcal{J}$ to represent product ownership; $c_j$ is the marginal cost of producing product $j \in \mathcal{J}$; $y$ is the representative consumer's income; $V_0 > 0$ is the market size parameter; $H_0 \geq 0$ is a baseline utility parameter\footnote{$\log H_0$ is interpreted as the value of the outside option. Typical empirical models set $H_0 = 1$. }; finally, $h_j: \mathbb{R}_{++} \to \mathbb{R}$ is defined as
\[
h_j(p_j) = 
\begin{cases}
\exp(v_j - \alpha p_j) & \text{(MNL)} \\
v_j p_j^{1-\sigma} & \text{(CES)}
\end{cases}
,
\]
where $\alpha > 0$ and $\sigma > 1$ are the price responsiveness parameters that determine the price elasticities and the substitutability of products.

Model primitives in \eqref{equation:model.primitives} determine consumer demand functions as follows. The representative consumer's quasilinear indirect utility function is $y + V(p)$, where $V(p) \equiv V_0 \log H(p)$ and $H(p) \equiv H_0 + \sum_{j \in \mathcal{J}} h_j(p_j)$. The function $V(p)$ measures consumer welfare at price $p$.  \citet{nocke2018multiproduct} refers to $H = H(p)$ as the industry (market) aggregator. Applying Roy's identity to the indirect utility function gives the demand function:
\begin{equation}\label{equation:roys.identity.demand.function}
    q_j(p) = V_0 \frac{-h_j'(p_j)}{H(p)}.
\end{equation}

 Demand functions in \eqref{equation:demand.functions} are consistent with \eqref{equation:roys.identity.demand.function} if $H_0 = 1$ and $V_0$ is given by \eqref{equation:monetary.scaling.factor}.

\subsection{Equilibrium Characterization}

\citet{nocke2018multiproduct} shows that type aggregation property of logit and CES demand systems can be used to summarize the equilibrium. The model primitives determine each firm's type as
\begin{equation}
    T_f \equiv 
    \begin{cases}
    \sum_{l \in \mathcal{J}_f} \exp(v_l - \alpha c_l) & \text{(MNL)} \\
    \sum_{l \in \mathcal{J}_f} v_l c_l^{1-\sigma} & \text{(CES)}
    \end{cases}
    .
\end{equation}
With these demand systems, each firm finds it optimal to apply the same markup to all of its products. To characterize the Bertrand equilibrium, it is convenient to define ``$\iota$-markups'' as
\begin{equation}
    \mu_f \equiv 
    \begin{cases}
        \alpha (p_j - c_j) & \forall j \in \mathcal{J}_f \quad \text{(MNL)} \\
        \sigma \frac{p_j - c_j}{p_j} & \forall j \in \mathcal{J}_f \quad \text{(CES)}
    \end{cases}
    .
\end{equation}
The $\iota$-markups are proportional to the actual markups, either in levels (MNL) or percentages relative to price (CES). The market aggregator $H$ is defined as
\begin{equation}
    H \equiv 
    \begin{cases}
        1 + \sum_{j \in \mathcal{J}} \exp(v_j - \alpha p_j) & \text{(MNL)} \\
        1 + \sum_{j \in \mathcal{J}} v_j p_j^{1-\sigma} & \text{(CES)} 
    \end{cases}
    .
\end{equation}

Let $s^f = \sum_{j \in \mathcal{J}_f} s_j$ be the market share of firm $f$. \citet{nocke2018multiproduct} shows that the Bertrand equilibrium can be characterized as $(\{\mu_f\}_{f \in \mathcal{F}}, \{s_f\}_{f \in \mathcal{F}}, H)$ (a vector of $\iota$-markup, firm-level market shares, and a market aggregator) satisfying the following system of equations:
\begin{align}
    1 & = 
    \begin{cases}
        \mu_f \left( 1 - \frac{T_f}{H} \exp(-\mu_f) \right)  & \text{(MNL)} \\
        \mu_f \left( 1 - \frac{\sigma - 1}{\sigma} \frac{T_f}{H} (1 - \frac{\mu_f}{\sigma})^{\sigma - 1} \right) & \text{(CES)}
    \end{cases}
    , \quad \forall f \in \mathcal{F},
    \\
    s_f & =
    \begin{cases}
        \frac{T_f}{H} \exp \left(- \mu_f \right) & \text{(MNL)} \\
        \frac{T_f}{H} \left( 1- \frac{1}{\sigma} \mu_f \right)^{\sigma - 1} & \text{(CES)}
    \end{cases}
    , \quad \forall f \in \mathcal{F},
    \\
    1 & = \frac{1}{H} + \sum_{f \in \mathcal{F}} s_f .
\end{align}

A unique solution to this system of equations is guaranteed to exist. The equilibrium price of each product is determined by the equilibrium value of $\mu_f$. Firm-level profits can be expressed as
\begin{equation}
    \Pi_f = V_0(\mu_f - 1).
\end{equation}
The consumer surplus is determined by the value of the market aggregator:
\begin{equation}
    \mathit{CS}(H) = V_0 \log(H).
\end{equation}

\section{Merger Pass-Through Matrix Under Logit and CES \label{section:merger.pass.through.matrix.under.logit.and.ces}}

\subsection{Logit}

This section characterizes the pass-through matrix under the logit assumption. Recall that the merger pass-through matrix is defined as $M \equiv - \left(\frac{\partial h(p^\text{pre})}{\partial p} \right)^{-1} $, where $h(p^\text{post}) = 0$ characterizes the post-merger FOC; each $h_j(p)$ is normalized to be quasilinear in marginal cost. Let the merged firm's profit function be $\sum_{l \in \mathcal{J}_A \cup \mathcal{J}_B} (p_l - c_l) q_l$. For each $j \in \mathcal{J}_A \cup \mathcal{J}_B$, 
\[
\begin{split}
h_j(p) & = - q_j \left( \frac{\partial q_j}{\partial p_j} \right)^{-1} - (p_j - c_j) - \sum_{l \in \mathcal{J}_A \cup \mathcal{J}_B \backslash j} (p_l - c_l) \left(\frac{\partial q_l}{\partial p_j} \right) \left( \frac{\partial q_j}{\partial p_j} \right)^{-1}    \\
& = - s_j \left( \frac{\partial s_j}{\partial p_j} \right)^{-1} - (p_j - c_j) - \sum_{l \in \mathcal{J}_A \cup \mathcal{J}_B \backslash j} (p_l - c_l) \left(\frac{\partial s_l}{\partial p_j} \right) \left( \frac{\partial s_j}{\partial p_j} \right)^{-1}    \\
& = -s_j \left( \frac{-1}{\alpha s_j(1-s_j)} \right) - (p_j - c_j) - \sum_{l \in \mathcal{J}_A \cup \mathcal{J}_B \backslash j} (p_l - c_l) \left( \alpha s_j s_l \right) \left( \frac{-1}{\alpha s_j (1-s_j)} \right) \\
& = \frac{1}{\alpha (1-s_j)} - (p_j - c_l) + \sum_{l \in \mathcal{J}_A \cup \mathcal{J}_B \backslash j} (p_l - c_l) \frac{s_l}{1-s_j}.
\end{split}
\]
Note that I use $q_j = s_j N$ and $\frac{\partial s_j}{\partial p_k} = -\alpha (\mathbb{I}\{j=k\} - s_j)s_k$ (so that $\frac{\partial s_j}{\partial p_j} = -\alpha s_j(1-s_j)$ and $\frac{\partial s_k}{\partial p_j} = \alpha s_j s_k$ if $j \neq k$). 

The next step is to characterize $\frac{\partial h_j}{\partial p_k}$ for $j,k \in \mathcal{J}_A \cup \mathcal{J}_B$ so that I can construct the $(\mathcal{J}_A \cup \mathcal{J}_B) \times (\mathcal{J}_A \cup \mathcal{J}_B)$ matrix $\frac{\partial h(p)}{\partial p}$. First, since $\frac{d}{dx}\left( \frac{1}{1-x} \right) = \frac{1}{(1-x)^2}$, $\frac{\partial (p_j - c_j)}{\partial p_j} = 1$, and $\frac{\partial (p_l - c_l)}{\partial p_j} = 0$ if $l \neq j$,
\[
\begin{split}
\frac{\partial h_j}{\partial p_j} & = \frac{1}{\alpha} \frac{1}{(1-s_j)^2} \frac{\partial s_j}{\partial p_j} - 1 + \sum_{l \in \mathcal{J}_A \cup \mathcal{J}_B \backslash j} (p_l - c_l) \left( \frac{\partial s_l}{\partial p_j} \frac{1}{1-s_j} + s_l \frac{1}{(1-s_j)^2} \frac{\partial s_j}{\partial p_j} \right) \\
& = \frac{1}{\alpha} \frac{1}{(1-s_j)^2} (-\alpha s_j (1-s_j)) - 1 + \sum_{l \in \mathcal{J}_A \cup \mathcal{J}_B} (p_l - c_l) \left( \frac{\alpha s_l s_j }{1-s_j} + s_l \frac{1}{(1-s_j)^2} (-\alpha s_j (1-s_j)) \right) \\
& = - \frac{s_j}{1-s_j} - 1  \\
& = \frac{-1}{1-s_j}.
\end{split}
\]
Next, for $k \neq j$, 
\[
\begin{split}
    \frac{\partial h_j}{\partial p_k} & = \frac{1}{\alpha} \frac{1}{(1-s_j)^2} \frac{\partial s_j}{\partial p_k} + \sum_{l \in \mathcal{J}_A \cup \mathcal{J}_B \backslash j} \left( \frac{\partial (p_l - c_l)}{\partial p_k} s_l \frac{1}{1-s_j} + (p_l - c_l) \frac{\partial s_l}{\partial p_k} \frac{1}{1-s_j} + (p_l - c_l) s_l \frac{1}{(1-s_j)^2} \frac{\partial s_j}{\partial p_k} \right) \\
    & = \frac{s_j s_k}{(1-s_j)^2} + \sum_{l \in \mathcal{J}_A \cup \mathcal{J}_B \backslash j} \left( \mathbb{I}\{l = k\} \frac{s_l}{1-s_j} + (p_l - c_l) (-\alpha (\mathbb{I}\{l=k\} - s_k)s_l) \frac{1}{1-s_j} + (p_l - c_l) \frac{ \alpha s_l s_j s_k}{(1-s_j)^2} \right).
\end{split}
\]
The pre-merger margin for each product $j$ owned by firm $f$ can be calculated as $(p_j - c_j) = \frac{1}{\alpha(1-s_f)}$. Thus, the matrix $\frac{\partial h(p^\text{pre})}{\partial p}$ and thus the merger pass-through matrix can be calculated from the pre-merger shares $(s_j)_{j \in \mathcal{J}_A \cup \mathcal{J}_B}$ and the price responsiveness parameter $\alpha$.

In the special case of single-product firms, we can calculate the pass-through matrix as
    \[
    M = \frac{(1-s_A)^3(1-s_B)^3}{(1-s_A-s_B)(1-s_A-s_B+2s_As_B)} 
    \begin{bmatrix}
        \frac{1}{1-s_B} & \frac{s_As_B}{(1-s_A)^2(1-s_B)} \\
        \frac{s_A s_B}{(1-s_A)(1-s_B)^2} & \frac{1}{1-s_A}
    \end{bmatrix}
    .
    \]

\subsection{CES}

Under the CES demand assumption, the merger pass-through matrix is defined as $M \equiv - \left( \frac{\partial h (\tilde{p})}{\partial \tilde{p}} \right)^{-1} \bigg \vert_{\tilde{p} = \tilde{p}^\text{pre}}$, where $\tilde{p}_j \equiv \log p_j$, and each $h_j(\tilde{p})$ is normalized to be quasilinear in relative margins $m_j = \frac{p_j - c_j}{p_j}$. As shown in \citet{koh2024merger}, for each $j \in \mathcal{J}_A \cup \mathcal{J}_B$, 
\[
h_j(\tilde{p}) = - \varepsilon_{jj}^{-1} - m_j + (1 + \varepsilon_{jj}^{-1}) \sum_{l \in \mathcal{J}_A \cup \mathcal{J}_B \backslash j} m_l D_{j \to l}^R,
\]
where $\varepsilon_{jj} \equiv \frac{\partial q_j}{\partial p_j} \frac{p_j}{q_j}$ is the own-price elasticity of demand and $D_{j \to l}^R = - \frac{\partial R_l / \partial p_j}{\partial R_j/\partial p_j}$ is the revenue diversion ratio from product $j$ to product $l$. Under the CES demand assumption, the revenue diversion ratio and the own-price elasticity can be simplified as $D_{j \to l}^R = \frac{s_l^R}{1- s_j^R}$ and $\varepsilon_{jj} = (1-s_j^R)(1-\sigma) - 1$.

Let us characterize the matrix of $\frac{\partial h_j(\tilde{p})}{\partial \tilde{p}_k}$ for $j,k \in \mathcal{J}_A \cup \mathcal{J}_B$. First, 
\[
    \frac{\partial h_j}{\partial \tilde{p}_j} = - \left( \frac{\partial \varepsilon_{jj}^{-1}}{\partial \tilde{p}_j} \right) - \frac{\partial m_j}{\partial \tilde{p}_j} + \frac{\partial (1 + \varepsilon_{jj}^{-1})}{\partial \tilde{p}_j} \sum_{l \in \mathcal{J}_A \cup \mathcal{J}_B \backslash j}m_l D_{j \to l}^R + (1 + \varepsilon_{jj}^{-1}) \sum_{l \in \mathcal{J}_A \cup \mathcal{J}_B \backslash j} m_l \frac{\partial D_{j \to l}^R}{\partial \tilde{p}_j} .
\]
\citet{koh2024merger} shows that $\frac{\partial \varepsilon_{jj}^{-1}}{\partial \tilde{p}_j} = \varepsilon_{jj}^{-2} s_j^R(1-s_j^R)(1-\sigma)^2$, $\frac{\partial \varepsilon_{jj}^{-1}}{\partial \tilde{p}_k} = (-1) \varepsilon_{jj}^{-2} s_j^R s_k^R (1-\sigma)^2$, $\frac{\partial m_j}{\partial \tilde{p}_j} = 1 - m_j$, $\frac{\partial D_{j \to k}^R}{\partial \tilde{p}_j} = 0$, and $\frac{\partial D_{j \to k}^R}{\partial \tilde{p}_k} = (1 - \sigma) \left( \frac{s_k^R(1 - s_k^R)}{1-s_j^R} - \frac{(s_k^R)^2 s_j^R}{(1- s_j^R)^2} \right)$. Then,
\[
\frac{\partial h_j(\tilde{p})}{\partial \tilde{p}_j} =  -\varepsilon_{jj}^{-2} s_j^R (1-s_j^R)(1-\sigma)^2 - (1 - m_j) + \varepsilon_{jj}^{-2} s_j^R (1-s_j^R) (1-\sigma)^2 \sum_{l \in \mathcal{J}_A \cup \mathcal{J}_B \backslash j} m_l D_{j \to l}^R.
\]
Next, for $k \neq j$,
\[
\begin{split}
\frac{\partial h_j(\tilde{p})}{\partial \tilde{p}_k} & = - \left( \frac{\partial \varepsilon_{jj}^{-1}}{\partial \tilde{p}_k} \right) + \frac{\partial (1 + \varepsilon_{jj}^{-1})}{\partial \tilde{p}_k} \sum_{l \in \mathcal{J}_A \cup \mathcal{J}_B \backslash j} m_l D_{j \to l}^R + (1 + \varepsilon_{jj}^{-1}) \sum_{l \in \mathcal{J}_A \cup \mathcal{J}_B \backslash j} \left( \frac{\partial m_l}{\partial \tilde{p}_k} D_{j \to l}^R + m_l \frac{\partial D_{j \to l}^R}{\partial \tilde{p}_k} \right)    \\
& =- \left( \frac{\partial \varepsilon_{jj}^{-1}}{\partial \tilde{p}_k} \right) + \frac{\partial (1 + \varepsilon_{jj}^{-1})}{\partial \tilde{p}_k} \sum_{l \in \mathcal{J}_A \cup \mathcal{J}_B \backslash j} m_l D_{j \to l}^R + (1 + \varepsilon_{jj}^{-1}) \left( \frac{\partial m_k}{\partial \tilde{p}_k} D_{j \to k}^R + m_k \frac{\partial D_{j \to k}^R}{\partial \tilde{p}_k} \right)  \\
& + (1 + \varepsilon_{jj}^{-1}) \sum_{l \in \mathcal{J}_A \cup \mathcal{J}_B \backslash \{j,k\}} m_l \frac{\partial D_{j \to l}^R}{\partial \tilde{p}_k} \\
& = \varepsilon_{jj}^{-2} s_j^R s_k^R(1-\sigma)^2 + (-\varepsilon_{jj}^{-2} s_j^R s_k^R (1-\sigma)^2) \sum_{l \in \mathcal{J}_A \cup \mathcal{J}_B \backslash j} m_l \frac{s_l^R}{1- s_j^R} \\
& + (1 + \varepsilon_{jj}^{-1}) \left( (1-m_k) \frac{s_k^R}{1-s_j^R} + m_k (1-\sigma) \left( \frac{s_k^R(1-s_k^R)}{1-s_j^R} - \frac{(s_k^R)^2 s_j^R}{(1-s_j^R)^2} \right) \right) \\
& + (1 + \varepsilon_{jj}^{-1}) \sum_{l \in \mathcal{J}_A \cup \mathcal{J}_B \backslash \{j,k\}} m_l (\sigma-1)\left( \frac{s_l^R s_k^R}{1-s_j^R} + \frac{s_l^R s_j^R s_k^R}{(1-s_j^R)^2} \right)
\end{split}
\]
Note that $\frac{\partial D_{j \to l}^R}{\partial \tilde{p}_k} = (\sigma - 1) \left( s_l^R s_k^R \frac{1}{1 - s_j^R} + s_l^R(1-s_j^R)^{-2} s_j^R s_k^R \right)$ for $l \neq j,k$ and that that the pre-merger relative margin $m_j = \frac{p_j - c_j}{p_j}$ for each product $j$ owned by firm $f$ can be calculated as $m_j = \left( \frac{1}{\sigma}\right) \left( \frac{1}{1 - \frac{\sigma - 1}{\sigma} s_f} \right) = \left( \frac{1}{\sigma} \right) \left( \frac{\frac{\sigma}{\sigma - 1}}{ \frac{\sigma}{\sigma - 1} - s_f} \right) $.\footnote{See the derivation in the Online Appendix A of \citet{caradonna2024mergers}.} Thus, the matrix $\frac{\partial h(\tilde{p}^\text{pre})}{\partial \tilde{p}}$ and thus the merger pass-through matrix can be calculated from the pre-merger shares $(s_j)_{j \in \mathcal{J}_A \cup \mathcal{J}_B}$ and the price responsiveness parameters $\sigma$.

\end{document}